\documentclass[12pt,journal,draftcls,letterpaper,onecolumn,twosided]{IEEEtran}
\usepackage{amsmath,amssymb,times}
\usepackage{multirow,epsfig,cite,psfrag}
\usepackage{times, amsmath}
\usepackage{stfloats}
\usepackage{dsfont}
\usepackage{color}
\usepackage{enumerate}
\usepackage{amsfonts}
\usepackage{amsthm}
\usepackage{amsmath}
\usepackage{array}
\usepackage{bbm}
\usepackage{cite}
\usepackage{dsfont}
\usepackage{epsfig}
\usepackage{float}
\usepackage[T1]{fontenc}
\usepackage{graphicx}
\usepackage{indentfirst}
\usepackage{multicol}
\usepackage{subfigure}
\usepackage{url}
\usepackage{multirow}
\usepackage{color}
\usepackage{epstopdf}
\usepackage{algpseudocode}
\usepackage[normalem]{ulem}
\usepackage{float}
\usepackage[nolist]{acronym}
\usepackage{booktabs}
\usepackage{subfigure}
\usepackage{graphicx}
\usepackage[keeplastbox]{flushend}

\usepackage{graphicx}
\usepackage{caption}
\usepackage{algorithm}
\usepackage{algpseudocode}

\newtheorem{theorem}{\bf Theorem}

\newtheorem{proposition}{\bf Proposition}
\newtheorem{lemma}{\bf Lemma}

\DeclareMathOperator{\EX}{\mathbb{E}}

\begin{document}
\title{\fontsize{19}{22} \selectfont Contextual Bandit Learning for Machine Type Communications in the Null Space of Multi-Antenna Systems}
	\author{\IEEEauthorblockN{Samad Ali, \emph{Student Member, IEEE}, Hossein Asgharimoghaddam, \emph{Student Member, IEEE}, Nandana Rajatheva,  \emph{Senior Member, IEEE}, Walid Saad, \emph{Fellow, IEEE}, and Jussi Haapola, \emph{Member, IEEE}}
	\thanks{A preliminary version of this work appeared in the IEEE  VTC Fall 2017 \cite{samadVTC}.}
	\thanks{S. Ali, H. Asgharimoghaddam, N. Rajatheva and J. Haapola are with the Centre for Wireless Communications (CWC), University of Oulu, Finland. Emails: \{samad.ali, hossein.asgharimoghaddam, nandana.rajatheva, jussi.haapola\}@oulu.fi. W. Saad is with Wireless@VT, Bradley Department of Electrical and Computer Engineering, Virginia Tech, Blacksburg, VA, USA, Email: walids@vt.edu.}}
\maketitle
\vspace{-2cm}
\begin{abstract}
Ensuring an effective coexistence of conventional broadband cellular users with machine type communications (MTCs) is challenging due to the interference from MTCs to cellular users. This interference challenge stems from the fact that the acquisition of channel state information (CSI) from machine type devices (MTD) to cellular base stations (BS) is infeasible due to the small packet nature of MTC traffic. In this paper, a novel approach based on the concept of opportunistic spatial orthogonalization (OSO) is proposed for interference management between MTC and conventional cellular communications. In particular, a cellular system is considered with a multi-antenna BS  in which a receive beamformer is designed to maximize the rate of a cellular user, and, a machine type aggregator (MTA) that receives data from a large set of MTDs. The BS and MTA share the same uplink resources, and, therefore, MTD transmissions create interference on the BS. However, if there is a large number of MTDs to chose from for transmission at each given time for each beamformer, one MTD can be selected such that it causes almost no interference on the BS. A comprehensive analytical study of the characteristics of such an interference from several MTDs on the same beamformer is carried out. It is proven that, for each beamformer, an MTD exists such that the interference on the BS is negligible. To further investigate such interference, the distribution of the signal-to-interference-plus-noise ratio (SINR) of the cellular user is derived, and, subsequently, the distribution of the outage probability is presented. However, the optimal implementation of OSO requires the CSI of all the links in the BS, which is not practical for MTC. To solve this problem, an online learning method based on the concept of contextual multi-armed bandits (MAB) learning is proposed. The receive beamformer is used as the context of the contextual MAB setting and Thompson sampling: a well-known method of solving contextual MAB problems is proposed. Since the number of contexts in this setting can be unlimited, approximating the posterior distributions of Thompson sampling is required. Two function approximation methods, a) linear full posterior sampling, and, b) neural networks are proposed for optimal selection of MTD for transmission for the given beamformer. Simulation results show that is possible to implement OSO with no CSI from MTDs to the BS. Linear full posterior sampling achieves almost $90\%$ of the optimal allocation when the CSI from all the MTDs to the BS is known.

\end{abstract}
\begin{IEEEkeywords} Machine type communications, scheduling, fast uplink grant, multi-armed bandits, internet of things, multi-antenna communications, deep contextual bandits, Thompson sampling. \end{IEEEkeywords}

\section{Introduction} \label{sec:introduction}
Conventional wireless communications systems are designed with the goal of providing high data rates for human type users. With the introduction of the Internet-of-Things (IoT), wireless networks should provide a new type of connectivity, known as machine-type-communications (MTC). In human type communications, applications require large data exchanges such as multimedia services or web browsing. In contrast, MTC applications that rely on uplink transmission of short packets, such as smart meters, environment monitoring, and factory automation, should be supported \cite{dawyM2MMagazine}. MTC is an essential part of the development of the next generation of cellular networks \cite{walid6G, IoTin5G}. However, there are fundamental differences between human type communications and MTC, which create wireless challenges that primarily arise from the short packet nature of MTC traffic, uplink centric data transmission, and heterogeneous quality-of-service (QoS) requirements of IoT applications such as latency and security \cite{aidin_tcom}. Meanwhile, there is a need to design systems that can provide connectivity for a massive number of machine-type-devices (MTD) \cite{MassiveM2M}, which is known as massive MTC. A major MTC challenge is that the data packets are small, and, therefore, the signaling overhead associated with sending scheduling requests and channel state information (CSI) acquisition is not negligible compared to the packet size. As a result, the optimal utilization of resources, especially for uplink scheduling of MTDs and the use of multi-antenna systems becomes very challenging.
\subsection{State-of-the-art}
There has been a surge of interest in the literature \cite{surveyofaccess, RACHM2M2,rach_correlated,nora, Samad-fastuplinkgrant, samad_globecom, DI_letter, NBLTE-M,3GPP-NB-IoT, CapillaryHamid, CapillaryMain} that has recently sought to address the MTC challenge that originate from the short packet nature of IoT traffic. One of the approaches is optimizing the random access process to reduce the signaling overhead (e.g., see \cite{surveyofaccess, RACHM2M2}, and references therein).  For instance, in \cite{rach_correlated}, the authors propose to exploit the correlations between the traffic patterns of different MTDs to optimize the random access process for MTC. Non-orthogonal random access is proposed in  \cite{nora} to increase the efficiency of scheduling request transmissions from MTDs. The works in \cite{Samad-fastuplinkgrant} and \cite{samad_globecom} discuss notion of a fast uplink grant in which the MTDs do not send scheduling requests and the base station (BS) allocates resources to MTDs. The fast uplink grant requires source traffic prediction \cite{DI_letter} which is the prediction of the set of active MTDs at any given time. The work in \cite{park} studies the use of learning to improve MTD scheduling in presence of urgent messages. Moreover, to design MTD-specific technologies, solutions such as narrow-band IoT (NB-IoT) \cite{3GPP-NB-IoT} and LTE-M \cite{NBLTE-M} allocate a part of wireless spectrum only for MTC. Another major method for increasing efficiency of wireless networks for MTC is by capillary networks \cite{CapillaryMain}, \cite{CapillaryHamid}. In capillary networks, MTC nodes transmit data to a machine-type-aggregator (MTA) and then MTA forwards collected data to the cellular network. Capillary networks help in better utilization of resources and use smaller distances between MTDs and MTA to save energy. Capillary networks can also be used in combination with multiple-input multiple-output (MIMO) systems where multiple antennas increase the capacity of wireless networks by providing spatial diversity and degrees-of-freedom (DoF) gain \cite{tse2005fundamentals}.

The use of a MIMO system for MTC faces serious challenges since MIMO systems require CSI to be known at transmitter and/or receiver. To investigate MIMO for MTC, the authors in \cite{massivemimomtc01} study grant-free MTC in massive MIMO and propose methods for user activity detection and channel estimation. Moreover, in \cite{massivemimomtc02}, the authors analyze achievable rates under maximal ratio combing (MRC) and minimum mean squared error (MMSE) receivers using random matrix theory. The work in \cite{randomBeamforming} studies the performance of random beamforming for MTC by considering that pilot signals are transmitted in the downlink and uplink while the downlink channel reciprocity is used for MTDs to decide on whether to transmit or not. In \cite{massivebeamforming}, the authors study the impact of beamforming in massive MTC by investigating the outage probability for various number of antennas at the BS. Clearly, the use of multiple antennas at the BS is beneficial for MTC, however, incorporating capillary networks in MIMO systems face a challenging problem of interference from MTDs to the cellular network. There has been little work on the study of interference management between MTC and human type communications in capillary networks.

One interesting approach to manage interference in MTC is the idea of opportunistic spatial orthogonalization (OSO). The concept of OSO takes advantage of the spatial diversity of received interference in the network. This scheme was initially introduced in \cite{shen2009dynamic, shen2011opportunistic, IntDraining} for cognitive radio networks, where there are a large number of secondary user transmit candidates. In each time slot, one of the secondary users is selected for transmission, such that interference from the selected secondary transmitter is minimum after being multiplied by receive beamformer of the primary link. The ideal scenario in OSO is that interference falls into the null space of the primary receiver's beamformer, so as to eliminate interference on the primary link. The main requirements for this technique are knowledge of CSI and the existence of a large number of secondary transmit candidates. This method is further developed in other works related to MIMO interference channels in \cite{perlaza2010spectrum}, \cite{junghoon2010}, \cite{lee2011interference}. The concept of OSO in conventional systems has two key drawbacks:
\begin{itemize}
	\item Considering a very large number of secondary users is not realistic in normal human type communications.
	\item The receive beamformer of the primary link changes at each time step. This means that the secondary user's opportunity to transmit can be lost in the next time step and it might not be able to transmit its data.
\end{itemize}
\subsection{Contributions}
Due to these limitations, it is not practical to implement OSO in conventional human type communications. However, this approach can be very suitable for MTC in capillary networks in which there are thousands of MTDs that want to transmit their data to an MTA. In MTC, the above-mentioned drawbacks do not exist due to the following reasons. First, the assumption of a very large number of users is realistic and natural for MTC. Second, MTDs transmit small packets and even a small transmit opportunity is enough for the MTD to transmit its data. We consider OSO for MTC where radio resources are shared MTC and human type cellular users in capillary networks. OSO is used for interference management where for each beamformer at the BS, one MTD is selected such that it causes no interference on BS. However, finding such an MTD is a challenge since it requires the CSI knowledge of all the links the BS and acquisition of such CSI is impractical. To implement OSO for MTC, we propose solution which exploits the channel characteristics of stationary MTD links using the well-known one-ring model \cite{tse2005fundamentals}. Moreover, an online learning method based on the multi-armed bandit (MAB) theory \cite{sutton1998reinforcement} is presented to select the best MTD for resource sharing. In particular, we use contextual bandits \cite{deep_contextual_bandits} which is a form of MAB learning where at each time a context is revealed to the learning agent, and, for each context, the best MTD is selected.

The main contribution of this paper is therefore to develop a novel framework, based contextual bandit learning for effectively leveraging OSO to enable coexistence between MTC and conventional cellular communications. Summary of the contributions of this paper are:
\begin{itemize}
    \item We present the idea of OSO for MTC and provide theoretical analysis to prove that, in MTC, it is possible to find an MTD that causes no interference to the BS.
	\item We study the effect of OSO on the performance of the human-type device by deriving the probability distribution function (PDF) of the signal-to-interference-plus-noise ratio (SINR) of the human type device. Moreover, we derive a closed-form expression for the outage probability of the human-type device as a function of on the number of MTDs in the system.
	\item Since it is impossible to have CSI from all the MTDs at the BS, we propose a novel learning approach based on contextual multi-armed bandits to find the best MTD for transmission. We first provide the channel characteristics based on the well-known one ring model for the human-type device and MTDs. Then, by using the generated channels, in our online learning approach, at each time step, we use the given receive beamformer in the BS as the context in our contextual bandit algorithm. 
	\item We propose a contextual bandit solution that is based on the well-known Thompson sampling method for MTD selection. We design the feature vector for the learning models by separating the real and imaginary parts of the receive beamformer. This feature vector is given the learning agent as the context in order to make the MTD selection process. Therefore, the learning agent acts as a function that maps every context to an MTD. Since the number of combinations of the feature vectors is unlimited, function approximation is proposed to model the mapping function. In the proposed framework, these function approximators are used to approximate the posterior distributions of Thompson sampling. First, we use a linear model which is based on linear regression learning. Second, we use neural networks with various sampling methods. We compare the performance of the proposed solutions to a baseline uniform sampling policy. Our proposed solution is then shown to enable the BS to select an MTD for resource sharing with the cellular user without any knowledge on the CSI of the MTD links. 
	\item Extensive simulations are carried out to evaluate the effectiveness of the proposed solutions in enabling OSO as a practical method for enabling coexistence of human-type and MTC communications.
\end{itemize}

The rest of the paper is organized as follows. Section \ref{sysmodel} presents the system model and problem formulation. In Section \ref{problem}, we introduce the idea of OSO for MTC, provide theoretical proof that, in MTC, it is possible to find an MTD for OSO, and study the performance of human type communications under interference from MTC. The contextual bandit learning framework is presented in Section \ref{deep_contextual}. Simulation results are presented in Section \ref{Simulationresults} and conclusions are drawn in Section \ref{conclutions}.

\section{System Model}\label{sysmodel}
Consider the uplink of a wireless cellular network having a single BS, a single human-type device, a set $\mathcal{K}$ of $K$ machine-type-device (MTDs), and, a machine-type-aggregator (MTA). The BS has $M$ antennas and the human-type device and MTDs are single antenna devices. The human-type device transmits data in the uplink to the BS and MTDs transmit data to the MTA. We consider an orthogonal frequency division multiple access (OFDMA) system where radio resources are divided into orthogonal resource blocks (RBs) in the time and frequency domains. These RBs are shared between MTDs and the human-type device which leads to an interference network where MTD and the human-type device transmissions cause interference at the BS and MTA, respectively. We assume that MTDs transmit in orthogonal RBs and only one MTD shares the RB with the human-type device. An illustrative system model is shown in Fig \ref{fig:sysmodel2}.

\begin{figure}[t]
 \begin{center}
   \includegraphics[width=0.5\textwidth]{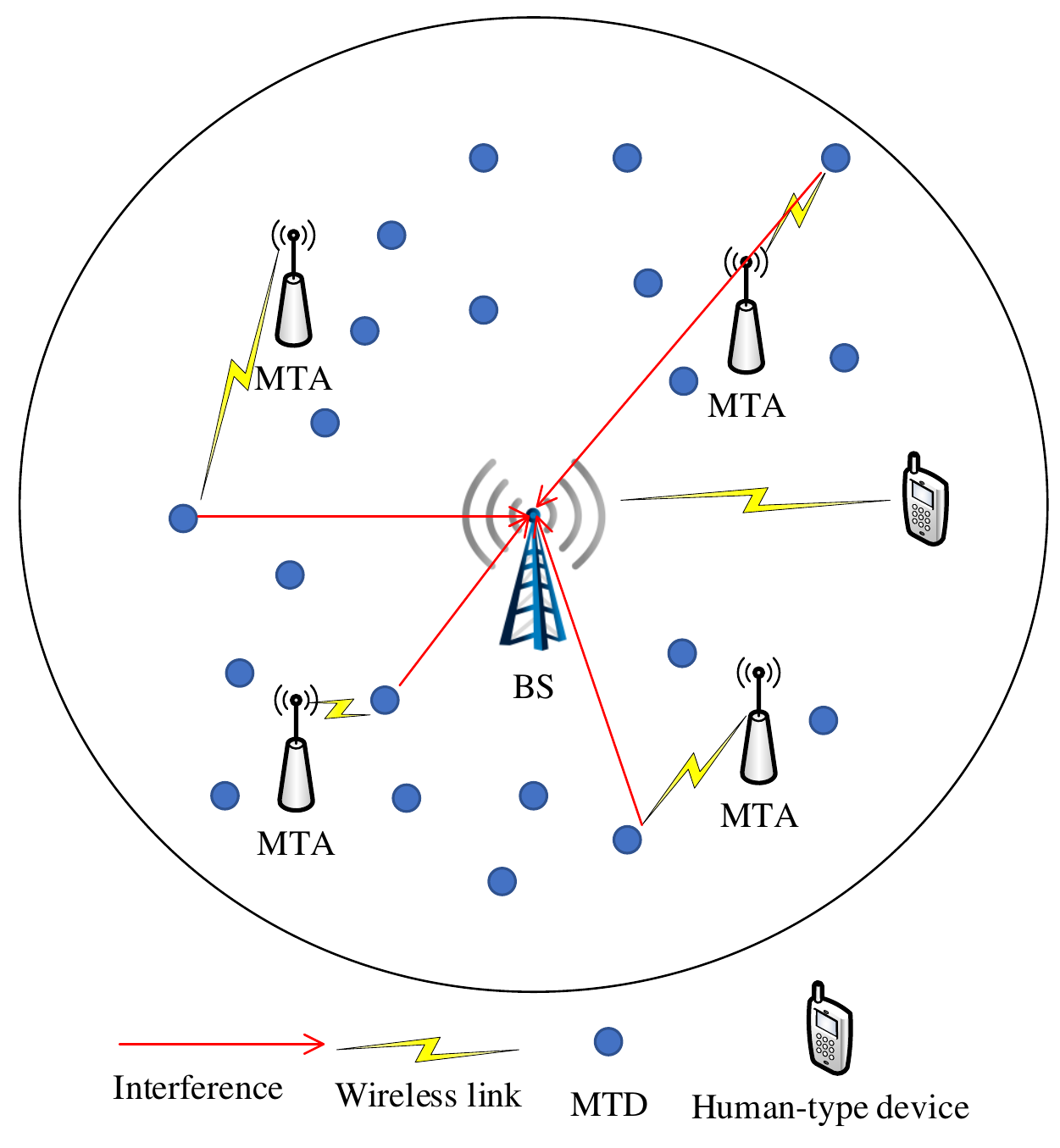}
 \end{center}
 \caption{\small System model}
 \label{fig:sysmodel2}
\end{figure}

At the BS, the received signal  $\boldsymbol y_{\rm BS}$ from the human-type device is given by:
\begin{equation}
\begin{aligned}\label{eq:RUplink}
&{\boldsymbol y}_{bs}& = {\boldsymbol h}_{c}x_{c} + \boldsymbol{h}_{k,B}x_{k} + {\boldsymbol n}_{BS},\\
\end{aligned}
\end{equation}
where  $x_{c}$ and $x_{m}$ are the transmit symbols of the human-type device and MTD respectively. $\boldsymbol{h}_{c} \in \mathbb{C}^M $ is human-type device to BS channel. MTD is selected from $K$ devices with channels $\boldsymbol{h}_{k,B} \in \{\boldsymbol{h}_{1,B},...,\boldsymbol{h}_{K,B} \} $. ${\boldsymbol n}_{\rm BS}$ represents the white Gaussian noise. After applying the receive beamformer at the BS, receiver has the following signal:
\begin{equation}
\begin{aligned}\label{eq:RUplinkbf}
&{\hat{\boldsymbol y}}_{bs}& = \boldsymbol{w}_{c}\boldsymbol{h}_{c}x_{c} + \boldsymbol{w}_{c}{\boldsymbol h}_{k,B}x_{k} + \boldsymbol{w}_{c}N_0.
\end{aligned}
\end{equation}
From \eqref{eq:RUplink} we can now write the SINR of the human-type device at the BS:
\begin{equation}
\begin{aligned}\label{eq:MBSSINR}
\gamma_{c} =  \frac{P_c |\boldsymbol{w}_{c}\boldsymbol{h}_{c}|^2}{P_k| \boldsymbol{w}_{c}\boldsymbol{h}_{k,B} |^2 + | \boldsymbol{w}_{c}|^2N_0 }.
\end{aligned}
\end{equation}
In a similar manner the received signal at the MTA can be written as follows:
\begin{equation}
\begin{aligned}\label{eq:RMTG}
y_{m} = {h}_{k}x_{k} + {h}_{c,M}x_{c} + {n}_{s},\\
\end{aligned}
\end{equation}
where ${h}_{k}$ is the channel between the MTD and MTA, and ${h}_{c,M}$ is the channel between the human-type device and the MTA. The Gaussian noise in MTA is ${n}_{\rm s}$.  The received SINR at MTA can be given by:
\begin{equation}
\begin{aligned}\label{eq:scSINR}
\gamma_{k} =  \frac{P_k|{h}_{k}  |^2}{P_c| {h}_{c,M} |^2 + N_0 },
\end{aligned}
\end{equation}
where $P_k$ and $P_c$ are, respectively, the MTD transmit power and the human-type device transmit power. Since we consider that the receive beamformer in the BS is designed to receive the data from the human-type device optimally, we assume that it can perfectly decode the data of the human-type device. We also assume that there is a backhaul link between the BS and the MTA, and therefore, the MTA can subtract the human-type device data from the received signal at the MTA using successive interference cancellation. Hence, we neglect the interference from the human-type device on the MTA and focus on the interference from MTDs on the BS. Clearly, if the receive beamformer $\boldsymbol{w_{\rm c}}$ is designed without taking the channel from MTDs into account, the interference on the BS can be severe. Moreover, taking the MTD transmission into account will require sending pilots to the BS to estimate the channel $\boldsymbol{h}_{kb}$. However, the amount of resources required for CSI estimation is large with respect to the data packets size of MTC applications. Moreover, it is not possible to estimate channels of all MTDs to allocate resources to only one of them. Therefore, CSI acquisition for MTDs is an inefficient method. Most importantly, even if the CSI is calculated perfectly at the BS, the BS must still use degrees of freedom to null out the interference from the MTD. This will lead to a decrease in the achievable rate for the human-type device at the BS. Another major issue is that the receive beamformer for the human-type device can be calculated once for each coherence period if the MTD is not taken into account. However, when considering the MTD, for each new MTD with a different channel, the receive beamformer for the human-type device must be calculated again to mitigate the interference from the new MTD. Therefore, designing a beamformer for a human-type device while considering the MTDs becomes very inefficient and suboptimal, and degrades the performance of the HTD transmissions. 

\begin{figure}[t]
	\begin{center}
		\includegraphics[width=0.5\textwidth]{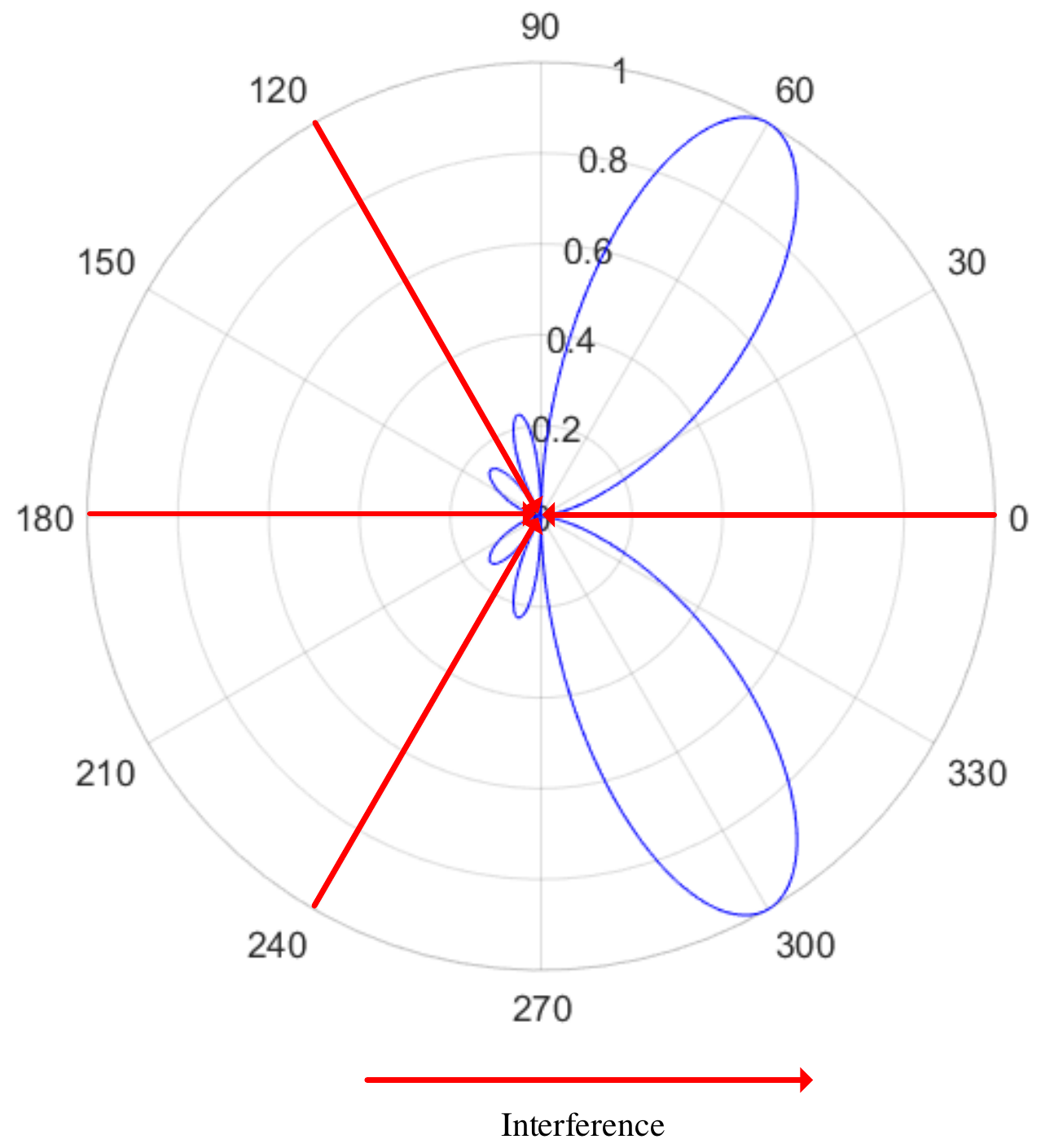}
	\end{center} 
	\caption{\small An example of a receive beamformer in polar coordinates that shows the direction for receiving the signal of interest and some example direction of interference that will not have a negative effect on the received SINR.}
	\label{fig:beamformer}
\end{figure}
\section{Opportunistic Spatial Orthogonalization for MTC}\label{problem}
We assume that the BS is responsible of the beamformer design for the human-type device and also for the scheduling of MTDs to transmit their data to the MTA. In the BS, for each coherence interval, a receive beamformer $\boldsymbol{w_{\rm c}} \in \mathbb{C}^M $ is designed without considering the interference from the MTD transmissions. The optimal beamformer for a single user SIMO system is maximal ratio combining (MRC) in which the beamformer is a normalized conjugate of the channel $\boldsymbol{h_{\rm c}} \in \mathbb{C}^M $ from the human-type device to the BS. For each beamformer $\boldsymbol{w}_{c}$, the BS can also select one MTD, $k \in \mathcal{K}$, such that received SINR of the uplink user is maximized. If the BS has the CSI from all the MTDs, this can be easily performed and the MTDs that cause a small amount of negligible interference on the BS can be scheduled. Furthermore, if the number of MTDs is very large, then the selected MTD might cause almost no interference to the BS at all since it can be selected such the interference originating from the selected MTD is in the null space of the receive beamformer at the BS. Next, we prove that, in wireless systems with a large number of MTDs that are \emph{interference candidates}, the MTD with minimum interference causes almost no interference to the BS. 

For our system model, when the MTD with minimum interference is scheduled to transmit using the same RBs that are allocated to the human-type device, the SINR of the human-type device will be:
\begin{equation}
\begin{aligned}\label{eq:MBSSINROSO}
\gamma_{c} =  \frac{P_c |\boldsymbol{w_{c}}\boldsymbol{h}_{\rm c}|^2}{\min_{k=1,...K}P_i| \boldsymbol{w_{c}}\boldsymbol{h}_{k,B} |^2 + | \boldsymbol{w_{c}} |^2N_0 }.
\end{aligned}
\end{equation}
As mentioned earlier, when there is one user and a BS with multiple antennas, the optimal receive method is MRC \cite{tse2005fundamentals} give by $\boldsymbol{w}_{c} = \boldsymbol{h_{\rm c}}^H$. Assume MTD $k$ as the one with minimum interference, the the SINR simplifies to:
\begin{equation}
\begin{aligned}\label{eq:MBSSINROSOMRC}
\gamma_{c} =  \frac{P_c \lVert \boldsymbol{h}_{c}\lVert^2}{\frac{P_k| \boldsymbol{h}_{c}^H\boldsymbol{h}_{kb} |^2}{\lVert{\boldsymbol h}_{\rm c}\lVert^2}  + N_0}.
\end{aligned}
\end{equation}
To further evaluate the performance of this scenario, we can write the probability of outage for the cellular user:
\begin{equation}
\begin{aligned}\label{eq:pout}
P_{out} =  \textrm{Pr}\left\{\frac{P_c \lVert{\boldsymbol h}_{c}\lVert^2}{\frac{P_k| \boldsymbol{h_{c}}^H\boldsymbol{h}_{kb} |^2}{\lVert \boldsymbol{h}_{c}\lVert^2}  + N_0} \leq \delta_{\textrm{th}} \right\}
\end{aligned}
\end{equation}
where $\delta_{th}$ is the maximum allowed outage probability. To avoid an outage, the interference term should be smaller than a threshold, show by $\delta_{\textrm{int}}$:

\begin{equation}\label{eq:QoS}
\begin{aligned}
\frac{P_k| \boldsymbol{h}_{c}^H{\boldsymbol h}_{k,B} |^2}{\lVert{\boldsymbol h}_{c}\lVert^2} < \delta_{\textrm{int}}
\end{aligned}
\end{equation}
If \eqref{eq:QoS} is satisfied, then, the selected MTD $k$ can transmit data to its MTA without harmful interference on the human-type device. In the following, we provide a theoretical analysis of the interference term in \eqref{eq:QoS} and show that, if the MTD is selected from a large set of candidates, \eqref{eq:QoS} is satisfied. In our analysis, the first step is to calculate the distribution of the interference term in the left-hand side of \eqref{eq:QoS}. To do that, we first prove that the interference from all possible MTD is independent of each other after going through the same receive beamformer. 

\subsection{Independence of Interference Candidates}
When one cellular user is transmitting to the BS and we want to select one out of $K$ MTDs for resource sharing, the interference term after being multiplied with receive beamformer appears in the denominator of \eqref{eq:MBSSINR}. We use random variables $X_{k} = \frac{ \boldsymbol{h}_{c}^H\boldsymbol{h}_{k,B} }{ |\boldsymbol{h}_{c}|}, k \in \mathcal{K}$ to capture the interference term. The following lemma holds:
\begin{lemma}\label{lemma1}
Random variables $X_k, k\in \mathbb{K}$ are independent.
\end{lemma}
\begin{proof}
By considering Rayleigh fading for channels between the MTDs and the BS, the elements of $\boldsymbol{h}_{k,B}$ will follow i.i.d, zero-mean, complex Gaussian distributions. The distribution of $X_{k}$  conditioned on $\boldsymbol{h}_{c}$ will also be complex Gaussian \cite{proakis2001digital}. The mean of the random variable $X_{k}$ conditioned on $\boldsymbol{h}_{c}$ will be:
\begin{equation}\label{mean}
\EX[X_{k}|\boldsymbol{h}_{c}] = \frac{\boldsymbol{h}_c^H}{|\boldsymbol{h}_{c}|} \EX[\boldsymbol{h}_{k,B}] = 0,
\end{equation}
subsequently, the variance is derived as follows:
\begin{align}\label{variance}
\EX[|X_{k}|^2|\boldsymbol{h}_{c}] & = \frac{\boldsymbol{h}_{c}^H \EX[\boldsymbol{h}_{k,B}\boldsymbol{h}_{k,B}^H]\boldsymbol{h}_{c}}{|\boldsymbol{h}_{c}|^2} \\
& = \frac{\boldsymbol{h}_{c}^H \textbf{I}_N \boldsymbol{h}_c}{|\boldsymbol{h}_{c}|^2} \\
& = 1.
\end{align}
A complex Guassian distribution with mean $\mu =0$ and variance $\sigma^2 = 1$ is standard complex normal distribution. Therefore, the distribution of $X_k$ conditioned on $\boldsymbol{h}_c$ has a PDF $f_{X_{k}}(X_k|\boldsymbol{h}_{c}) = \frac{1}{\pi} e^{-|X_k|^{2}}$ which is independent of $\boldsymbol{h}_{c}$. Given this independence, random variables $X_{k}$ are independent, and, hence, the interference which is $X_{k}^2$ term in \eqref{eq:MBSSINR} will be composed of independent random variables.
\end{proof}
Lemma \ref{lemma1} is used in the following theorem to prove that the probability of finding an MTD that will cause almost no interference on the BS becomes one as the number of MTDs increase.
\begin{theorem}\label{theorem1}
The probability to find an MTD that will satisfy the criteria in (\ref{eq:QoS}) becomes one when $K \rightarrow \infty$.
\end{theorem}
\begin{proof}
We start the proof by finding the distribution of the interference term. Since the distribution of $\boldsymbol{h}_c$ is $f_{x_{i}}(x_i) = \pi^{-1} e^{-|x_i|^{2}}$, then the interference power $|x_i|^2$ follows a central chi-square distribution with two degrees of freedom (which is an exponential distribution). Let us define the distribution of $z_i = |x_i|^2$ so that the distribution of interference term will be:
\begin{equation}\label{distribution}
f_{z_{i}}(z_i) = \frac{1}{p_i} e^{-\frac{z_i}{p_i}}.
\end{equation}

We now use order statistics \cite{papoulis1965probability} to find the distribution of the minimum term of the interference. Consider random variable $Y=\min\{z_1,z_2,...,z_K\}$. It is easy to show that the cumulative distribution function (CDF) of the minimum of a group of i.i.d exponential random variables is given by:
\begin{equation}\label{mincdf}
F_{Y}(y) = 1 - [e^{-\lambda Ky}],
\end{equation}
PDF of $Y$ follows as:
\begin{equation}\label{minpdf}
f_{Y}(y) = \lambda Ke^{-\lambda Ky}.
\end{equation}
From this distribution, it is clear that as $K \rightarrow \infty$, the probability, $P(Y < \delta_{\textrm{min}}) = 1$. 

This completes the proof.
\end{proof}
From Theorem \ref{theorem1}, we can see that whenever the network has a large number of MTDs that can be selected for resource sharing with an human-type device, if the MTD with minimum interference power on the BS is selected, the human-type device and MTD can transmit at the same time and on the same frequency band, such that the interference on the HTD is negligible. We should note that, for large $K$, Theorem \ref{theorem1} also holds for $K-1$, $K-2$, and $K-l$ MTDs, where $l$ is a small number. This means that $l$ resource blocks allocated for the human-type device can be shared with MTDs, while the interference on the human-type device is negligible. 

\subsection{Distribution of SINR and Outage Probability}
To study the performance of the proposed OSO for MTC, we derive the closed form expressions for the distribution of the SINR and the outage probability of the HTD while the MTD with minimum interference on HTD is transmitting on the same RB.
\begin{proposition}\label{prpopsition1}
The distribution of the SINR in (\ref{eq:MBSSINROSO}) is:
\begin{equation}\label{SINIR_distribution}
f_y(y)=\frac{\lambda e^{\lambda \sigma^2}}{P^M\Gamma(M)} \frac{y^{(M-1)}}{(\lambda + y/p)^{M+1}}\Gamma (M+1, (\lambda + y/p)\sigma^2),
\end{equation}
where $\Gamma(x,y)$ is upper incomplete gamma function.
\end{proposition}
\begin{proof}
The SINR is the ratio of two terms. The nominator of the SINR equation is Gamma distributed with PDF:
\begin{equation}\label{nominator}
f_x(x) = \frac{1}{P^M\Gamma(M)}x^{M-1}e^{-(x/P)},
\end{equation}
and the nominator of the SINR has an exponential distribution with PDF:
\begin{equation}\label{denom}
f_w(w) = \frac{K}{P_m}e^{-(K/P_m)(w-\sigma^2)} = \lambda e^{-\lambda (w-\sigma^2)}.
\end{equation}
Let $y$ be a random variable that represents the SINR. Since we have $y=x/w$, from the ratio of two random variables \cite{papoulis1965probability}, we can write the PDF of $y$ as:
\begin{equation}\label{ratio}
f_y(y) = \int_{\sigma ^2}^{\infty}w f_{x,w}(yw,w)dw.
\end{equation}
Since the terms in the nominator $f_x(x)$ and the denominator $f_w(w)$ of the SINR are independent, we have:
\begin{equation}\label{mutual}
f_{x,w}(x,w) = f_x(x) f_w(w),
\end{equation}
therefore, the PDF of the SINR is derived from:
\begin{equation}\label{ratio2}
f_y(y) = \int_{\sigma ^2}^{\infty}\frac{w}{P^M\Gamma(M)}(yw)^{M-1}e^{-(yw/P)}\lambda e^{-\lambda (w-\sigma^2)} dw,
\end{equation}
which can be computed as follows:
\begin{align}\label{sinrdistribution}  \nonumber
f_y(y) & = \frac{\lambda e^{\lambda \sigma^2}}{P^M\Gamma(M)} y^{(M-1)} \int_{\sigma ^2}^{\infty}w^M e^{-(yw/P)}  e^{-\lambda w} dw\\ \nonumber
&\!=\!\frac{\lambda e^{\lambda \sigma^2} }{P^M\Gamma(M)}y^{(M-1)} \int_{\sigma ^2}^{\infty}w^M e^{-(\lambda + y/p)w} dw \\
&\!=\!\frac{\lambda e^{\lambda \sigma^2}}{P^M\Gamma(M)}\frac{y^{(M-1)}}{(\lambda + y/p)^{M+1}}\Gamma (M+1, (\lambda + y/p)\sigma^2).
\end{align}
where $\Gamma(x,y)$ is the incomplete Gamma function. This completes the proof.
\end{proof}
The SINR distribution can be used to derive the outage probability of the proposed method. The derived outage probability distribution can be used to provide an analytical evaluation of the performance of the proposed OSO method. Such an analytical result can provide us with the number of MTDs that are needed in the MTC system so that we can select the best MTD for transmission. In the following theorem, the derivation of the outage probability is presented.
\begin{theorem}\label{theorem2}
For any give SINR threshold $\beta$, the outage probability is given by:
\begin{align}\label{outage} \nonumber
P(P_y(y)\leq \beta) = F_y(\beta)= &  1 -  \sum_{k=0}^{M-1} \frac{1}{P^K k!} \lambda e^{\lambda \sigma^2} \frac{\beta^k}{(\lambda + \frac{\beta}{P})^{k+1}} \Gamma(k+1,(\lambda + \frac{\beta}{P})\sigma^2) 
\end{align}
\end{theorem}
\begin{proof}
Outage happens when a received SINR is below a defined threshold. The probability of the outage for the SINR term in \eqref{SINIR_distribution} is then calculated as follows:
\begin{equation}\label{out1}
P(y\leq \beta) =\int_{0}^{\beta} f_y(y)dy.
\end{equation}\label{cdf1}
by plugging $f_y(y)$ from \eqref{SINIR_distribution} we have:
\begin{align}\label{cdfderivation}  \nonumber
F_y(\beta) &= \int_{0}^{\beta} \frac{1}{P^M\Gamma(M)}\lambda e^{\lambda \sigma^2} y^{(M-1)} \int_{\sigma ^2}^{\infty}w^M e^{-\lambda w} e^{-(\frac{y}{p})w}dwdy\\ 
&\!=\!\frac{1}{P^M\Gamma(M)} \lambda\!e^{\lambda \sigma^2} \int_{\sigma ^2}^{\infty} w^M e^{-\lambda w} \int_{0}^{\beta} y^{(M-1)} e^{-(\frac{w}{P})y}dydw
\end{align}
By using the integration rules from \cite[page 340]{integrals_books} we can derive the integral as follows:
\begin{align}\label{cdfderivation2}\nonumber
F_y(\beta)  & = \frac{1}{P^M\Gamma(M)}\lambda e^{\lambda \sigma^2} \int_{\sigma ^2}^{\infty}\!w^M e^{-\lambda w} \int_{0}^{\beta} y^{(M-1)} e^{-(\frac{w}{P})y} dydw \\\nonumber
& = \frac{1}{P^M\Gamma(M)}\lambda e^{\lambda \sigma^2}  \int_{\sigma ^2}^{\infty}  w^M e^{-\lambda w} \Bigg( \frac{(M-1)! P^M}{w^m} -e^{-(\frac{\beta}{P})w} \sum_{k=0}^{M-1} \frac{(M-1)! P^M}{k!} \frac{\beta^k}{P^k w^{M-k}} \Bigg) dw \\\nonumber
& = \lambda e^{\lambda \sigma^2} \int_{\sigma ^2}^{\infty} w^M e^{-\lambda w}\!\Bigg( \frac{1}{w^M} - e^{-(\frac{\beta}{P})w} \sum_{k=0}^{M-1}  \frac{\beta^k}{k!P^k\!w^{M-k}}\!\Bigg)\!dw \\\nonumber
& = \lambda e^{\lambda \sigma^2} \Bigg(\int_{\sigma^2}^{\infty} e^{-\lambda w}\!dw - \int_{\sigma^2}^{\infty} w^m e^{-\lambda w} e^{-(\frac{\beta}{P})w} \sum_{k=0}^{M-1}\!\!\frac{\beta^k}{k!P^k w^{M-k}} dw \Bigg)  \\\nonumber
& = 1 - \lambda e^{\lambda \sigma^2}  \int_{\sigma^2}^{\infty} e^{-(\lambda + \frac{\beta}{P})w} \sum_{k=0}^{M-1} \frac{1}{k!} \frac{\beta^k w^{k}}{P^k} dw \\\nonumber
& = 1 - \lambda e^{\lambda \sigma^2} \sum_{k=0}^{M-1} \frac{1}{k!} (\frac{\beta}{P})^k \int_{\sigma^2}^{\infty} e^{-(\lambda + \frac{\beta}{P})w}w^{k}  dw \\
& = 1 - \lambda e^{\lambda \sigma^2} \sum_{k=0}^{M-1} \frac{1}{k!} (\frac{\beta}{P})^k (\lambda + \frac{\beta}{P})^{-k-1} \Gamma(k+1,(\lambda + \frac{\beta}{P})\sigma^2) 
\end{align}
therefore, the probability of outage is given by:
\begin{align}\label{outageinproof} \nonumber
P(y\leq \beta) = F_y(\beta)= 1 -  \sum_{k=0}^{M-1} \frac{1}{P^K k!} \lambda e^{\lambda \sigma^2}  \frac{\beta^k}{(\lambda + \frac{\beta}{P})^{k+1}}  \Gamma(k+1,(\lambda + \frac{\beta}{P})\sigma^2) 
\end{align}
This completes the proof.
\end{proof}
Throughout this section, have presented the idea that an MTD can transmit in the null-space of HTD transmissions without causing interference. However, to design an optimal system that can exploit this characteristic of the system, the CSI from all the MTDs to the BS must be known. Indeed, as outlined earlier since the signaling overhead makes CSI acquisition inefficient for short data packet MTC traffic and estimating channels from many MTDs to allocate resources to one of them is highly inefficient. Therefore, in the next section, we present a contextual MAB learning framework to implement OSO for MTC. We exploit the characteristics of the channel between MTDs and the BS. Then, this characterization is used for generating channels that we use in our learning algorithm for finding the best possible MTD for any given beamformer in the BS. 

\section{Exploiting Interference Diversity}\label{deep_contextual}
In this section, we provide a practical method to find the set of MTDs that fall into the null spaces of each receive beamformer at the BS. Our method uses two concepts. First, since most of the BSs are located at high altitude, there are mostly not many scatterers around the BS. This means that a signal coming from a specific angle will have an angle of arrival that is limited to a short range of angles in that specific direction. If the scatterers were around the BS, the angle of arrival would be distributed uniformly in $[0,2\pi]$. Second, we use deep learning for finding the set of best possible MTDs for each designed beamformer.

\subsection{Channel characteristic and the angle of arrival of the interference}
In wireless systems, several different copies of the signal that is transmitted reach the receiver from different angles. The angle of arrival of the signal to the receiver depends on the reflectors around the receive antenna and also the transmit antenna. When the scatterers are around the BS, then the angle of arrival of the received signal is distributed uniformly in $[0,2\pi]$. However, if the scatterers are located in a ring around the transmitter, then the angle of arrival of the signal to the receiver is within a small interval. This is presented in Fig. \ref{anglefigure} and is known as the small ring model. In cellular systems, since most of the MTDs are located in buildings and locations with many scatterers around them, then the signal is mostly reflected as shown in Fig. \ref{anglefigure}. However, since the BS in most cases is located at a high altitude (e.g., on top of a building), there are fewer reflectors, and therefore, the received signal's angle of arrival is in a small interval. Similar channel characteristics exist for the human-type device. Next, we first present how the channel is calculated from the propagation environment characteristics. Then, we propose our learning based approach for scheduling the MTDs. 

\subsubsection{Channel Modeling}
Consider a cellular device that is located in distance $d$ from the BS with an angle $\theta$ as shown in Fig. \ref{anglefigure}. The angular spread of the signal $\Delta$ is calculated from $\Delta \simeq \arctan(r/s)$ where $r$ is the radius of the ring around the cellular device that scatterers are located. We assume that there is no line of sight and therefore, the channel $\boldsymbol{h} \sim \mathcal{CN}(\boldsymbol{0}, \boldsymbol{R}_k)$ is calculated as follows:
\begin{equation}\label{eq:channelDecomposition}
	\boldsymbol h = \boldsymbol U \boldsymbol \Lambda^{\frac{1}{2}} \boldsymbol w
\end{equation}
where $\boldsymbol U$ is the tall unitary matrix of the non-zero eigenvalues of $\boldsymbol{R}_k$, $\boldsymbol{\Lambda}$ is the $r \times r$ diagonal matrix of whose diagonal elements are the non-zero eigenvalues of $\boldsymbol{R}_k$ and $\boldsymbol{w} \sim \mathcal{CN}(\boldsymbol{0}, \boldsymbol{I})$. Each element $(m,p)$ of the the covariance matrix for the channel is calculated from:
\begin{equation}\label{eq:covarianceMatrix}
	\left[\boldsymbol{R}\right]_{m,p} =  \frac{a_i}{2\Delta}\int_{-\Delta}^{\Delta} e^{-j\boldsymbol{k}^T{(\alpha + \theta)}\left(\boldsymbol{u}_m - \boldsymbol{u}_p\right)}d\alpha
\end{equation}
where $\boldsymbol{k}(\alpha) = -\frac{2\pi}{\lambda}(\cos(\alpha), \sin(\alpha))^T$ is the wave vector for a planar wave impinging AoA $\alpha$, the carrier wavelength is $\lambda$ and $\boldsymbol{u}_m$ and $\boldsymbol{u}_p$ are the positions of the antennas of the BS in the two dimensional coordinate systems.  Path loss and shadowing are included in $a_i = 10^{\frac{a_{i,dB}}{10}}$ where $a_{i,dB} = PL_{dB} + X_\sigma$ with $PL_{dB}$ and $X_\sigma$ denoting the path loss and log-normal shadowing with variance $\sigma$. We use the 3GPP path loss model from the BS to MTDs \cite{3GPPPathLossModel} which is given by $PL_{\textmd{dB}} = 128.1 + 37.6\log(d)$. By using \eqref{eq:channelDecomposition} and \eqref{eq:covarianceMatrix}, we can derive the channel for each link the system. For uniform linear array (ULA) which is a widely used model, \eqref{eq:covarianceMatrix} simplifies to:
\begin{equation}\label{eq:covarianceMatrixULA}
\left[\boldsymbol{R}\right]_{m,p} =  \frac{a_i}{2\Delta}\int_{-\Delta}^{\Delta} e^{-j2\pi{(m - p)}\sin(\alpha)}d\alpha.
\end{equation}

\begin{figure}[t]
	\begin{center}
		\includegraphics[width=0.5\textwidth]{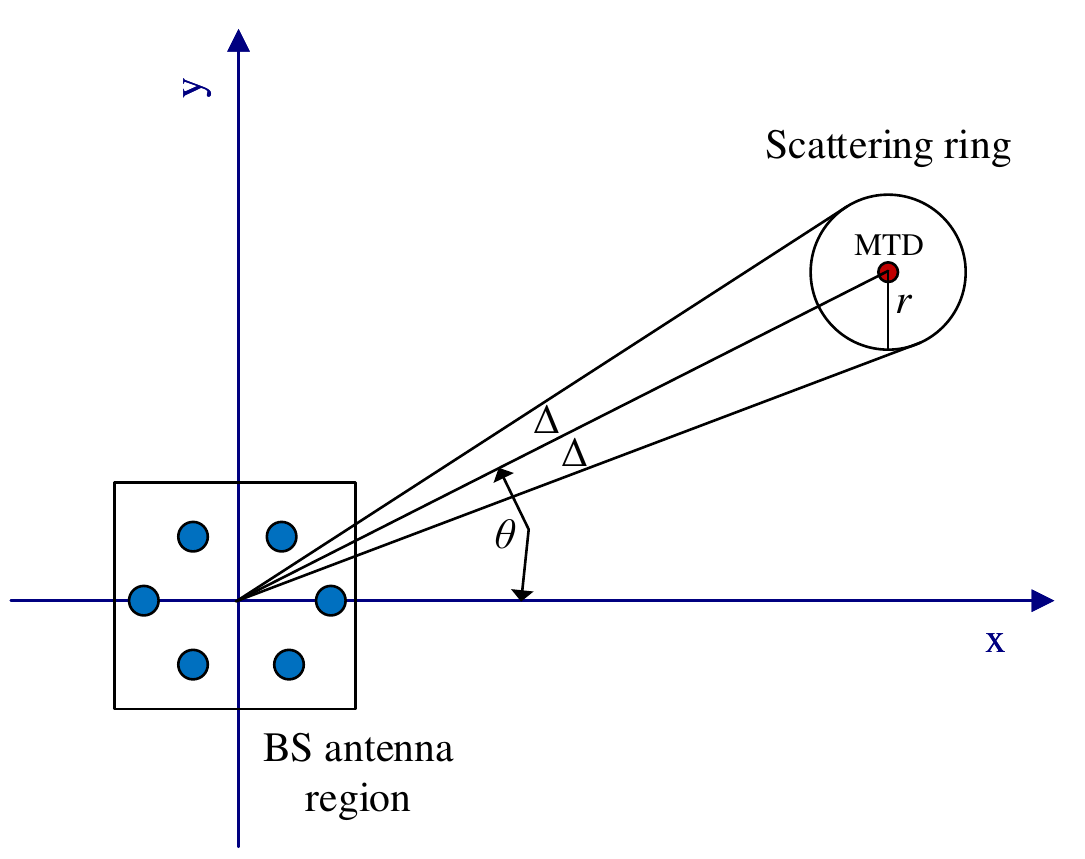}
	\end{center}
	\caption{\small Angle of arrival of the signal to the BS}
	\label{anglefigure}
\end{figure}

For MTDs and the human-type device, since the scatterers are located around the transmitter as shown in Fig. \ref{anglefigure}, we consider a random angle of arrival for each MTD and the HTD. We assume that each MTD has angle of arrival of $\theta_i$ and angular spread $\left[\theta - \zeta, \theta + \zeta\right]$. We assume that this interval fo AoA does not change for each MTD since the MTDs are either fixed or low mobility. However, during each coherence period, the AoA for the human-type device changes and at the beginning of the coherence period, it is estimated by using pilot signals. Therefore, this channel $\boldsymbol{h}$ in \eqref{eq:channelDecomposition} can be exactly calculated for human-type device at the beginning of each coherence period. This signaling for channel estimation is not performed for MTDs since they have small data packets and signaling is not efficient. Therefore, MTD channel knowledge can only be statistical, and, due to the existence of the noise term $\boldsymbol{w}$ in \eqref{eq:channelDecomposition}, the exact knowledge of the channel cannot be known for the MTDs. Therefore, we need to find a solution for selecting the best MTD for each given beamformer at the BS for human-type device without precise knowledge of MTD channels. In the following, we present a method based on contextual bandit theory to find the best MTD for transmission for each beamformer that is designed using the estimated channel of human-type device.

\subsection{Contextual Bandits}
In this section, we propose a learning method for selecting the best MTD for each receive beamformer by using methods from machine learning. Known as contextual bandits, our proposed solutions is an special case of MABs when the learning agent receives a context at each time, and, based on that context, selects an action. In the following, we first present MAB theory and the concept of contextual MABs. Then, we present two function approximations for contextual MAB problem of selecting best MTDs for resource sharing with human-type devices.
\subsubsection{Multi-armed bandits}
Reinforcement learning (RL) problems with a single state are called MABs. In an MAB problem, a decision maker (player) selects (plays) an arm from a set of available arms and receives a numeric reward for the selected arm. Rewards of each arm are drawn from a random variable with an underlying distribution that is not known to the decision maker. At each play, only the reward of the select arm is revealed to the player. The goal of the player is the minimize the cumulative regret over a long period of time. Regret is defined as the difference between the arm that is selected and the best possible arm could have been played. Total regret for the duration of $T$ of the learning period is given as:
\begin{equation}\label{regret01}
R(T) = \mathbb{E}\bigg [ \sum_{t=1}^{T}\theta^*(t) - \sum_{t=1}^{T}\theta_{k}(t) \bigg ],
\end{equation}
where the expectation is taken over the randomness in the algorithm and the revealed rewards. Naturally, the player will find the arm with the highest expected value and keep playing that arm. In a MAB problem, the main issue to resolve is the dilemma between exploration and exploitation. That is, to select the best arm that is known so far (exploitation) or play other suboptimal arms (explore) to have a better estimate of their reward distribution. There are many effective methods that can be used to solve exploration and exploitation dilemma such as $\epsilon-\text{greedy}$ or upper confidence bound (UCB) methods \cite{sutton1998reinforcement}. In our problem, each arm is an MTD that will be selected to be scheduled and the reward is the rate of the cellular user since each MTD transmission will cause a different amount of interference on the BS, and, therefore, affect the rate of the cellular user. Our problem has a major difference with the classical MAB problem described above. In the classical MAB, rewards of each arm are drawn from a fixed distribution, and, the best arm doesn't change over time. However, in our problem formulation, depending on the beamformer that is designed for the cellular user, the best arm will be different. Therefore, there is a \emph{context} that affects the best MTD (the context is the receive beamformer in the BS). This kind of problems is known as contextual MABs \cite{contextual_mab}. In contextual MABs, there is a need for a function $f(x)$ that gets the context as input and produces the best arm for that given input. Therefore, the aim of the learning is the find such a function that minimizes the regret over the learning period. We must state that the logarithmic regret is desirable in MAB as it means that at each new time step, the learning agent is making a smaller mistake. 

One of the well-known methods for solving contextual MAB problem is the so-called Thompson sampling method \cite{thompsonSampling}. In this method, for each action $a \in \mathcal{A}$, a distribution with parameters $\theta \in \Theta$ is considered that generates the reward $(r | a; \theta; x)$ where $x \in \mathcal{X}$ is the context from the set $\mathcal{X}$ of contexts, and $\Theta$ is the set of parameters for the distributions of the rewards. There is also a tuple of previous actions and rewards $\mathcal{D} = (x; r; a)$ that is available for the learning agent. The posterior distribution is $P(\theta | \mathcal{D}) \simeq P(\mathcal{D}| \theta) p(\theta)$ and the aim of the Thompson sampling method is to to maximize the expected rewards 
\begin{equation}\label{thompsonsampling}
\int \mathbb{I} \left[\mathbb{E}(r| a^{*}, x, \theta) = \max_{a^{\prime}} \mathbb{E}(r | a^{\prime}, x, \theta) \right] P(\theta | \mathcal{D}) d\theta
\end{equation}
where $\mathbb{I}$ is the indicator function. By using $P(\theta | \mathcal{D})$ at each time, parameters $\theta^{*}$ are sampled and then the action $a^{*}$ that maximizes $\left[\mathbb{E}(r| \theta^{*}, a^{*}, x) \right]$ is selected. In other words, the algorithm first considers a distribution for each arm and draws a sample from these distributions. Then, the agent acts greedily and selects the arm with the highest sample value and observes the reward. Based on the observed reward, parameters of these distributions are then updated and in the next round, the samples are drawn from the updated distributions. For each given context, this process is done separately, and, the aim is to solve estimate the parameters of these posterior distributions for new contexts. For solving a contextual MAB problem with a small number of contexts, we can run one learning algorithm for each context and store the learning information in tables. However, when the number of possible contexts is large (especially when it can have continuous values), then the problem becomes much more complicated, and, therefore, the function $f(x)$ must be approximated. This means that once the algorithm is trained on some contexts and actions, for any new context, it must be able to select the optimal action. In the next section, we present learning based function approximators for solving the contextual MAB problem based on Thompson sampling. 

\begin{algorithm}[t]
	\caption{Contextual Bandits}
	\begin{algorithmic}
		\State Receive prior distribution over models, $\pi_{o} : \theta \rightarrow [0,1]$
		\State Select each action $a$ once regardless of the context $\boldsymbol{q_{\rm c}}(t)$
		\State \textbf{for} $t=1$ \emph{to} $T$ \textbf{do}
		\State \quad \quad Observe $\boldsymbol{q_{\rm c}}(t)$, $\mathcal{D} = (\boldsymbol{q_{\rm c}}(t); r; a)$
		\State \quad \quad Sample the model $\theta_{t} \in R^{d}$
		\State \quad \quad Find $a_{t}$ = $\arg$$\max$ $P(\theta_{t} | \mathcal{D})$
		\State \quad \quad Update the posterior distributions $\pi_{t+1}$ with the observation $(\boldsymbol{q_{\rm c}}(t); r; a)$
	\end{algorithmic}
	\label{alg:probabilisticAlgoirthm}
\end{algorithm}

\subsection{Learning models}
As explained earlier, we need to design function approximators to sample the posterior distributions of Thompson sampling. Clearly, the aim is to minimize the regret over the learning period. The performance of learning based function approximators heavily depend on the choice of features that are fed into the learning algorithm and the type of the function approximator. In our problem formulation, we use a linear regression model and neural networks as the function approximator of the contextual MAB problem. To select the features that are fed into the function approximator we use the normalized receive beamformer. Since learning algorithms are mostly designed to work with real numbers, we separate the real and imaginary parts of the beamformer vector to design our features. Considering $M$ antennas in the BS, we build the feature vector $\boldsymbol{q_{\rm c}}(t) \in \boldsymbol{R}^{2M}$ at time $t$ as follows:
\begin{equation}\label{featurevector}
\boldsymbol{q_{\rm c}}(t)= \begin{bmatrix} \Re(\boldsymbol{w_{\rm c}}(t)) \\ \Im(\boldsymbol{w_{\rm c}}(t)) \end{bmatrix} / |\boldsymbol{w}_{\rm c}(t)|^2,
\end{equation}
where $\Re(\boldsymbol{w_{\rm c}}(t)) $ and $\Im(\boldsymbol{w_{\rm c}}(t))$ are the real and imaginary parts of the receive beamformer $\boldsymbol{w_{\rm c}}$ vector at time $t$. The contextual MAB algorithm $\mathcal{A}$ receives the context  $\boldsymbol{q_{\rm c}}(t)$ at time $t$, and, based on the internal model of the function approximator, selects an MTD for resource sharing. The BS calculates the rate of the cellular user at the BS, and, produces the normalized reward $r(t)$. The algorithm then updates the internal model of the function approximator based on the new data. In the framework of Thompson sampling for MAB problems, there are several several approaches that can be used for approximating the posterior distributions that are used for sampling and generating the possible reward for each MTD \cite{deep_contextual_bandits}. Moreover, these methods perform differently depending on the task at hand and analytical evaluation of the performance of these approaches is almost impossible and out of the scope of this work. We have used the implementation of all the approaches presented in \cite{deep_contextual_bandits} and selected several policies for sampling the data in our contextual MAB setting and a baseline uniform sampling policy. These methods are chosen due to their performance in our initial simulations both in terms of regret and computational time. The details of these methods are outlined as follows:

\subsubsection{Linear function approximation}
Consider the linear regression where the rewards for each arm is generated based on the following equation:
\begin{equation}\label{linearmodel}
r = \boldsymbol{q_{\rm c}}(t)^{T} \beta + \epsilon,
\end{equation}
where $r$ the reward that is generated, $\epsilon \sim \mathcal{N}(0, \sigma^2)$ and $\beta$ show parameters of the model that generate the reward for given context $\boldsymbol{q_{\rm c}}(t)$. The joint distribution of $\beta$ and  $\sigma^2$ are modeled which leads to sequentially estimating the noise level  $\sigma^2$ for each action. This leads to adaptive improvement of the parameters of the posterior distributions. 

The joint distribution of  the posteriors is given by $\pi_{t}(\beta, \sigma^2) = \pi_{t}(\beta | \sigma^2)\pi_{t}(\sigma^2)$. For the noise we consider the inverse Gamma distribution $\sigma^2 \sim IG(a_{t}, b_{t})$ and Gamma distribution for $\beta | \sigma^{2} \sim \mathcal{N}(\mu_{t}, \sigma^{2}\Sigma_{t})$.  Parameters of these distributions are given as follows:
\begin{equation}
\begin{aligned}
\Sigma_{t} & = (X^{T}X + \Lambda_0)^{-1}\\
\mu_{t} & = \Sigma_{t}(\Lambda_0\mu_{0} + X^{T}Y)\\
a_t & = a_0 + t/2\\
b_t & = b_0 + \frac{1}{2}(Y^{T}Y \mu_{0}^{T}\Sigma_{0}\mu_{0} - \mu_{t}^{T} \Sigma_{t}^{-1} \mu_{t}) 
\end{aligned}
\end{equation}
The hyper parameters are initialized as $\mu_{0} = 0$, $\Lambda_0 = \lambda \boldsymbol{I}M$ and $a_0 = b_0 = \eta >1$. After initialization, the parameters of the posteriors are updated after selecting each MTD by using linear regression model.

\subsubsection{Deep Neural Networks}
The above presented method of approximating the contextual MAB function with a linear model has limitation as it cannot model non-linearity is the mapping function from contexts to actions. To deal with non-linear mappings, neural networks can be used as function approximators. Recently, several neural network models are explored in \cite{deep_contextual_bandits} for solving the contextual bandit problems by approximating the posterior distributions of Thompson sampling. We experiment with the methods proposed in \cite{deep_contextual_bandits} and present the results in Section \ref{Simulationresults}. Here, in neural linear model, a Bayesian linear regression is applied after the last layer of the neural network. This can help in capturing the benefits of the both of linear and neural network models. Moreover, another method which is known as dropout is also considered, in which the outputs of the neural networks are randomly set to zero with some probability during the learning period \cite{droput}. This method is shown to improve the performance of neural networks is various tasks. 

To summarize, to solve the problem problem of contextual bandit learning for MTD selection task with unlimited number of possible beamformers as contexts, we have to use function approximations. First, a dataset is generated that includes beamformers that are generated by using MRC of the human-type device channel to the BS. Second, by using \eqref{eq:MBSSINR} the SINR for the human-type use considering the interference from all the possible MTDs is calculated. The dataset includes the feature vector of \eqref{featurevector} and all the SINRs. Then, for training the learning models, samples are randomly selected and used for training.
 
\section{Simulation Results}\label{Simulationresults}
For our simulations, we consider LTE parameters and we assume at each time step, the resources of the cellular user are shared with one MTD. Resources are primarily allocated to the cellular user and we assume that the CSI is acquired at the BS to design the uplink receive beamformer. Next, we first present the results to show that if there are a large number of MTDs to select from, the best MTD will cause a small amount of interference on the BS. Then, we present the plots for the distribution of the SINR and the outage probability. Finally, the results of deep contextual bandits are presented. Simulation parameters are given in Table \ref{tab:simparameters}.
\begin{table}[t]
	\def\tablename{Table}
	\centering
	\caption{Simulation parameters.} 
	\begin{tabular}{|c|c|} 
		\hline 
		Parameter & Value\\
		\hline\hline
		Number of antennas in BS ($M$) & 4\\
		\hline
		Cell radius & 500 m\\ 
		\hline
		MTA radius & 250 m\\
		\hline
		Bandwidth & 360 kHz\\
		\hline
		Noise figure at BS and MTA & 2 dB\\
		\hline
		CU to BS path loss model  & 128.1 + 36.7log($d$[km])\\
		\hline
		CU target SINR & 10 dB\\
		\hline
		MTD target SNR & 10 dB\\
		\hline
		Noise spectral density & -174 dBm/Hz\\
		\hline
		Angular spread $\Delta$ & $10^{\circ}$\\
		\hline
		Log-normal shadow fading & 10 dB \\
		\hline
		Location of antennas (y-axis) & $[-0.02, -0.01, 0.01, 0.02]$\\
		\hline 
	\end{tabular}
	\label{tab:simparameters}
\end{table}
\subsection{Existence of the null-space}
First, we study the possible effect of the number of MTDs on the SINR of the cellular user. Two transmit power control mechanisms for MTDs are considered. First, fixed transmit power from MTDs, and, second, a transmit power control to satisfy the SNR requirements at the MTA. In our simulations the MTA is equipped with a single antenna, and, MTDs use a basic distance based power control mechanism. The target SINR for the human-type user is set to $10$ dB and sharing the radio resource between MTDs and the human-type user, the degradation of the SINR of the cellular user due to MTD interference, is presented in Fig. \ref{fig:result01} for the case with fixed transmit power for MTDs, and in Fig. \ref{fig:result02} for the scenario with power control for the MTDs. 

Fig. \ref{fig:result01} shows the effect of interference from MTDs on the SINR of the human-type device for a different number of MTDs. This figure gives a clear example of how a  large number of MTDs provides the diversity to select one MTD with very small interference on BS. In these results, we assume that the MTDs do not perform power control and transmit with a fixed power all the time. We observe that even for a large transmit power of $10$ dBm, if there are $200$ MTDs to select from, the interference at the BS can be kept at a minimal level. However, for small transmit powers such as $0$ dBm, less than $100$ MTDs are enough to have a good performance for the human-type device while an MTD is sharing the same resources. 
\begin{figure}[t]
 \begin{center}
   \includegraphics[width=0.6\textwidth]{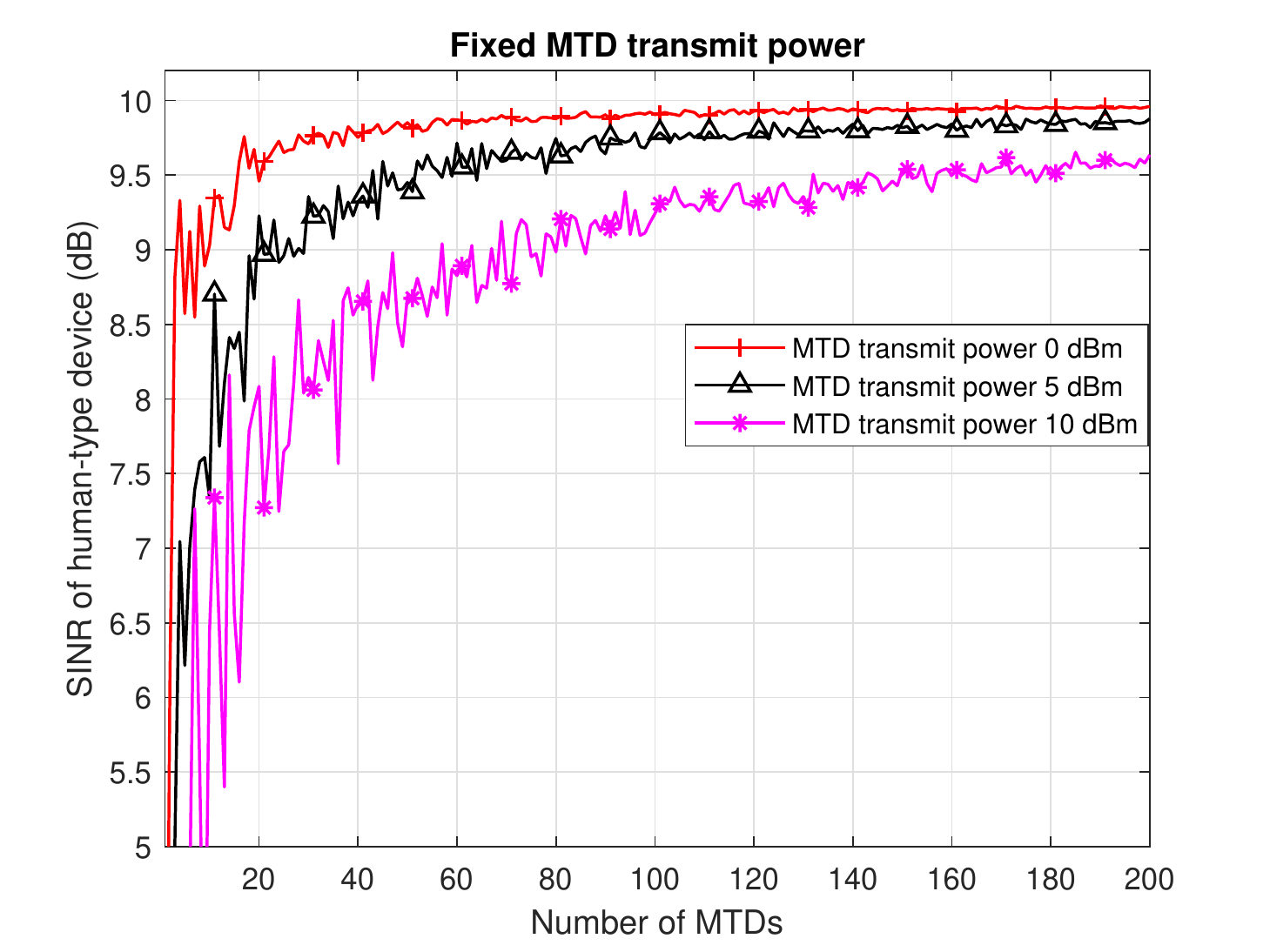}
 \end{center}
 \caption{\small SINR of the cellular user with fixed MTD transmit power.}
 \label{fig:result01}
\end{figure}

In Fig. \ref{fig:result02}, a similar plot to Fig. \ref{fig:result01} is given for the scenario of having power control for MTDs. Since the distance between each MTD and the MTA is typically small, MTDs transmit with low powers which leads to less interference. Clearly, higher SINR target values lead to higher transmit power, and, therefore, for larger values of SINR, a larger number of MTDs is required as selection candidates. Fig \ref{fig:result02} shows that for short distances between MTDs and the MTA, small transmit powers are needed, and, therefore, a smaller number of transmission candidates are needed in the system to find an MTD that causes no interference on the BS. Figs. \ref{fig:result01} and \ref{fig:result02} show that a large number of MTDs in MTC will make it possible to exploit the availability of null-space in the multi-antenna system for opportunistic interference management.

\begin{figure}[t]
 \begin{center}
   \includegraphics[width=0.6\textwidth]{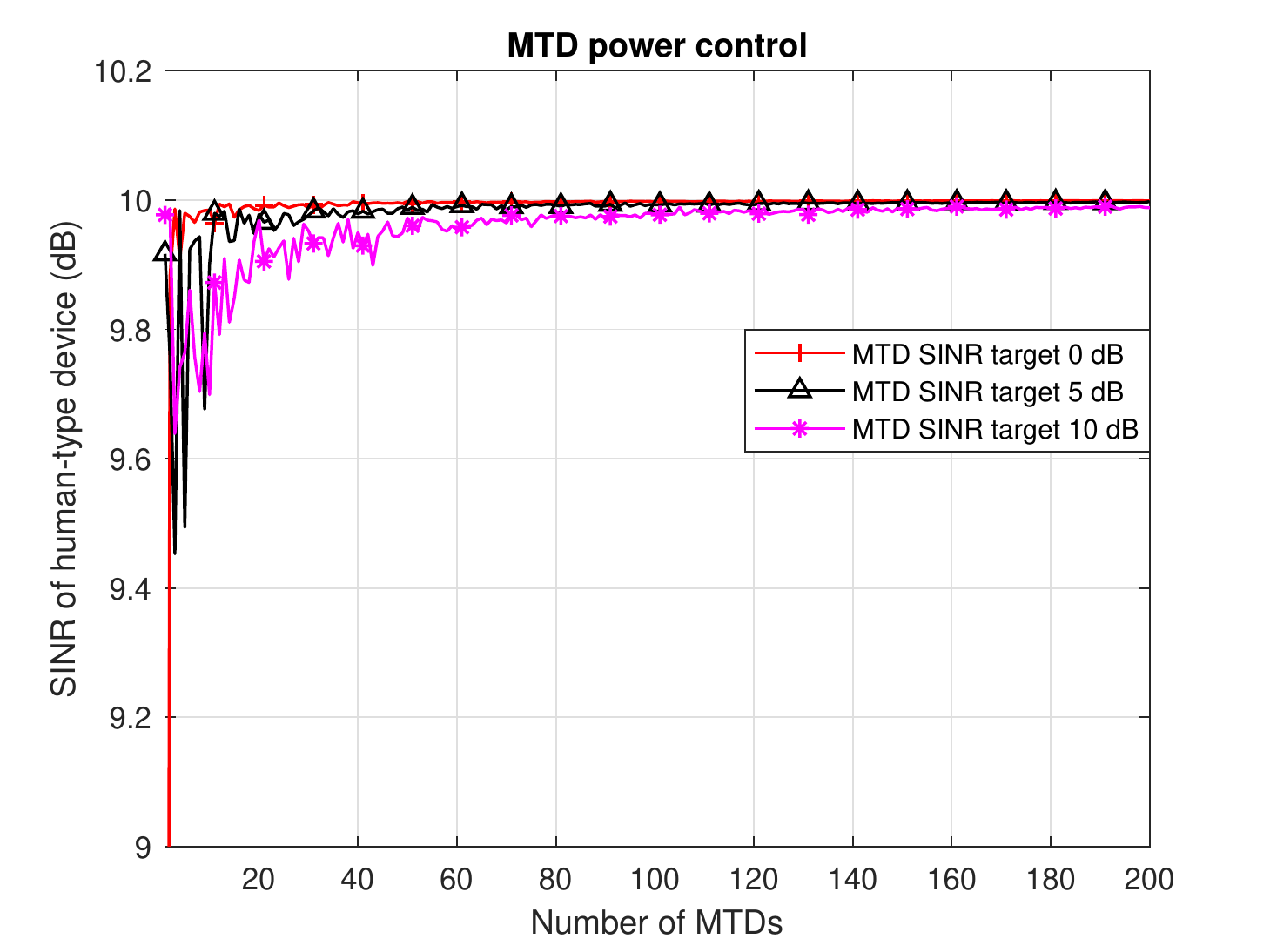}
 \end{center}
 \caption{\small SINR of the cellular user with MTD transmit power control.}
 \label{fig:result02}
\end{figure}

\subsection{Outage of human-type device}
Next, we present the simulation results for the outage probability of the cellular user in Fig. \ref{fig:ourage}. The outage probability is very small for a large number of MTDs and it is a decreasing function of number of MTDs. This validates the analytical results of Theorem \ref{theorem1} in Section \ref{problem}. Moreover, Fig. \ref{fig:ourage} shows that the idea of utilizing OSO for interference management in MTC can be realized by having around $100$ MTDs in the system. We also observe that with approximately $100$ MTDs in the system as possible interference candidates, the BS will most likely not be in the outage if the best MTD is selected. Moreover, this result helps us in selecting the number of required MTDs in the numerical analysis of the contextual bandits since we need to have a fixed number of MTDs for training the learning models. 

\begin{figure}[t]
	\begin{center}
		\includegraphics[width=0.6\textwidth]{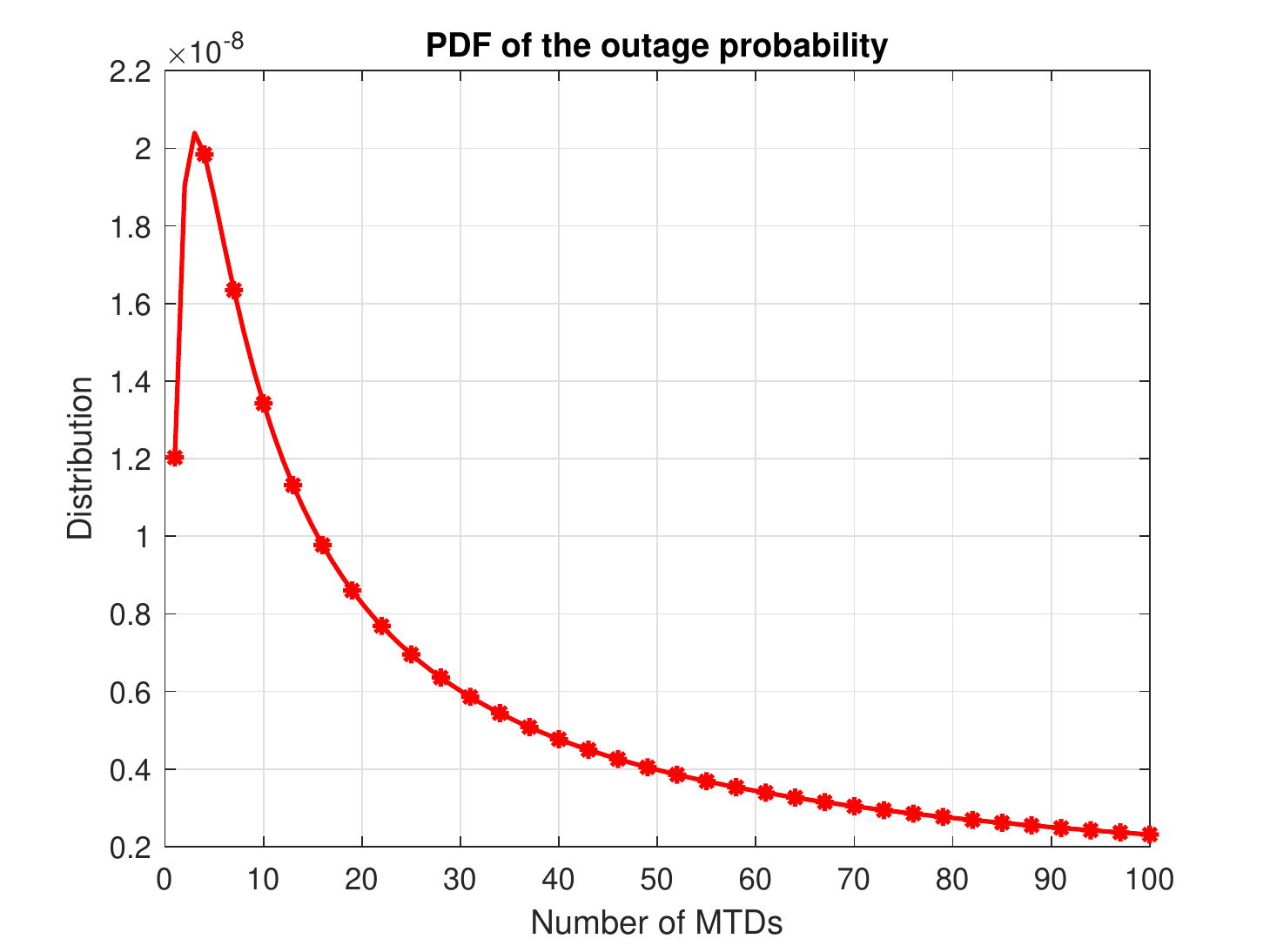}
	\end{center}
	\caption{PDF of the outage probability for different number of MTDs as transmission candidates.}
	\label{fig:ourage}
\end{figure}

\subsection{Contextual bandits}
Here, we present the results of the proposed contextual bandits for selecting the MTD with no CSI at the BS. The neural network that is used for approximating the function $f(x)$ two hidden layers. The input layer as $8$ neurons for a system with $4$ antennas. Each hidden layer is composed of $120$ neurons and the output layer has $80$ neurons where each neuron corresponds to one MTD. We run the learning algorithm for $20000$ iterations. 

\begin{table}[t]
	\def\tablename{Table}
	\caption{Contextual bandit rewards}
	\centering 
	\begin{tabular}{|c|c|}
		\hline 
		Algorithm & Cumulative reward\\
		\hline\hline 
		Optimal policy (full CSI) & 19996.04\\
		\hline
		Linear full posterior & 18043.06\\ 
		\hline
		Neural linear & 12639.06\\
		\hline
		Dropout & 8658.19\\
		\hline
		Uniform sampling & 7601.28\\
		\hline
	\end{tabular}
	\label{tab:deepbanditresults}
\end{table}
In Table \ref{tab:deepbanditresults}, we present the total reward received by our proposed methods compared to the optimal policy. The optimal policy is the one where the full CSI of all link was known at the BS. We can see that the linear full posterior is able to achieve to $90\%$ of the optimal policy. Clearly, this is a very good performance for a beamforming system with no CSI known at the BS. We can see that uniform sampling has very poor performance compared to linear full posterior. The uniform sampling is when the MTDs are selected uniformly at random. The results in Table \ref{tab:deepbanditresults} show that the contextual bandits with linear function approximation is capable of learning the various aspects of the wireless channel statistics between the MTDs and the BS. In other word, the contextual bandit is able to learn the angle of arrival of the interference from the MTDs and their signal strength, and, for any given beamformer, select the MTD with minimum interference on the BS.

In Fig. \ref{fig:regret}, we present the regret of the proposed policies. From Fig. \ref{fig:regret} we can clearly see that uniform sampling leads to linear regret. In MAB problems, sub-linear regret is desirable as it means that the number of times that the MAB algorithm selects a sub-optimal MTD  becomes smaller at each new time step. Dropout has poor performance as seen from the results in Table \ref{tab:deepbanditresults}, however, it starts to have a logarithmic regret after $14,000$ iterations. Logarithmic regret in MAB settings is considered as the optimal result \cite{sutton1998reinforcement}. Neural linear also achieves logarithmic regret, however, the total achieved rewards is not close to optimal policy. This shows that for the function approximation in our contextual MAB setting, neural networks do not perform well. Clearly, as seen from Table \ref{tab:deepbanditresults} and Fig. \ref{fig:regret}, linear full posterior sampling achieves very low regret and near-optimal results. 

\begin{figure}[t]
	\begin{center}
		\includegraphics[width=0.6\textwidth]{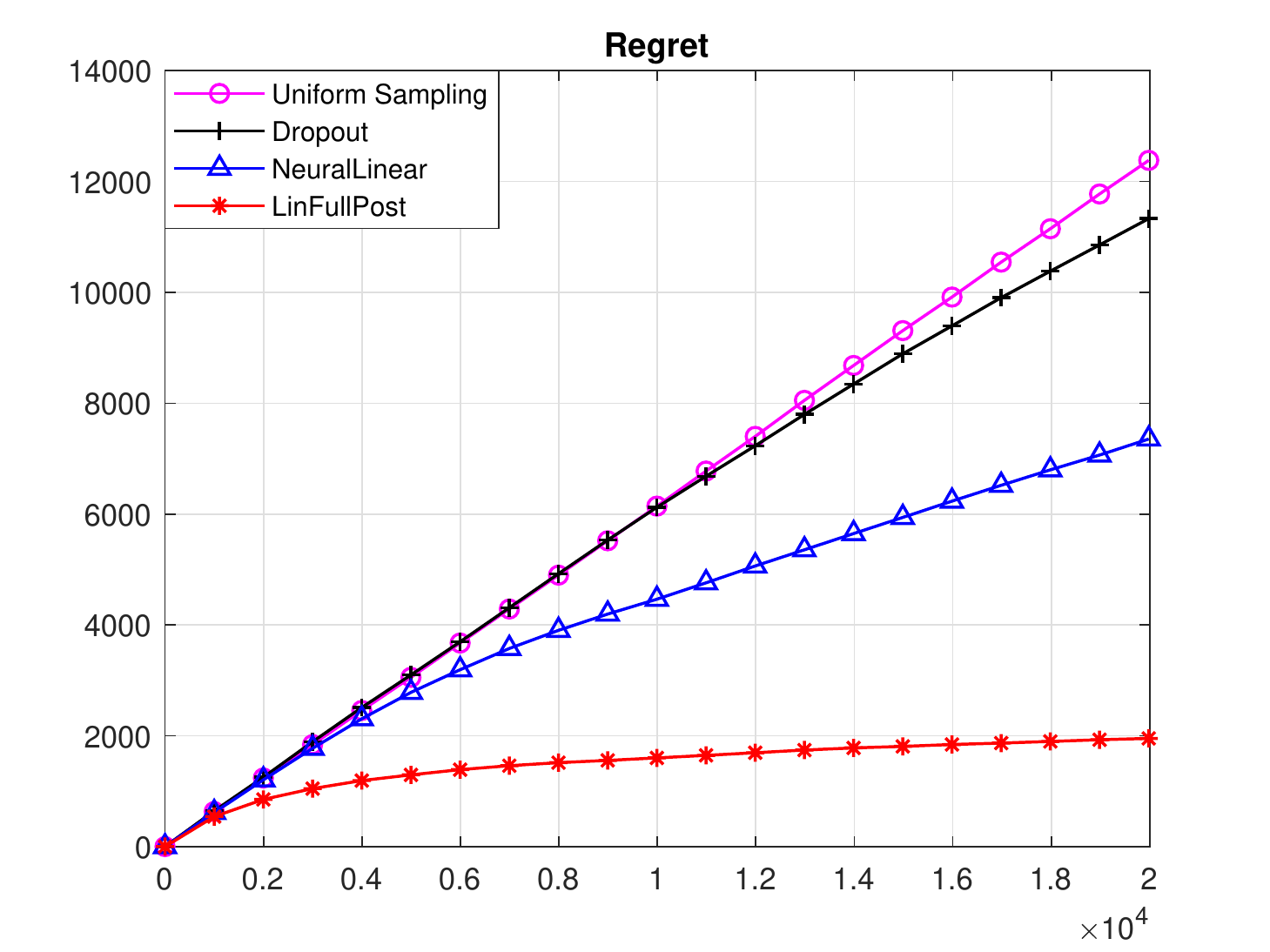}
	\end{center}
	\caption{\small Regret of deep contextual bandits compared to uniform sampling linear full posterior sampling.}
	\label{fig:regret}
\end{figure}

\begin{figure}[t]
	\begin{center}
		\includegraphics[width=0.6\textwidth]{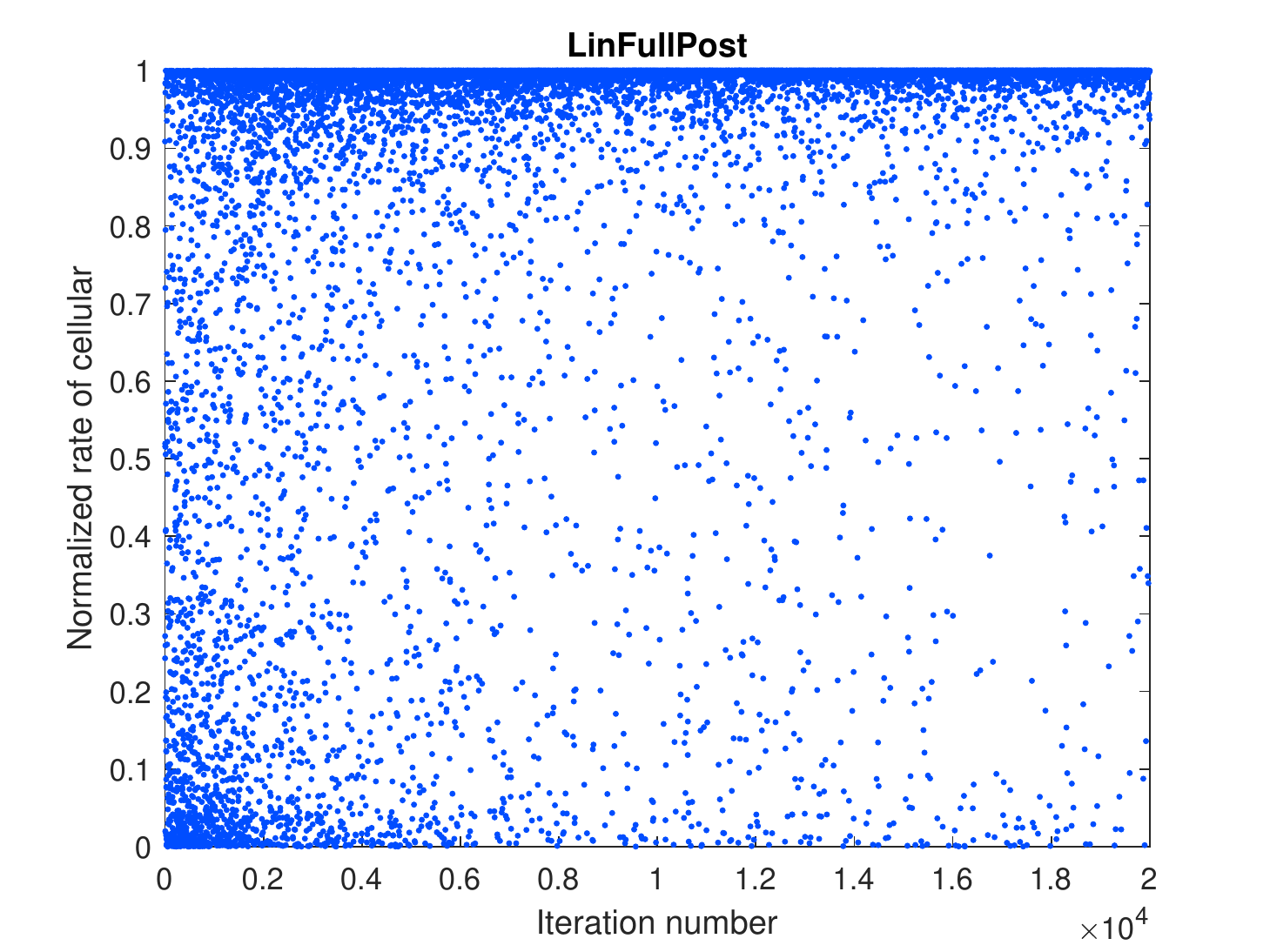}
	\end{center}
	\caption{\small Scatter plot of the normalized rate of the cellular user at each time step of the learning period.}
	\label{fig:scaterlinfullpost}
\end{figure}

\begin{figure}[t]
	\begin{center}
		\includegraphics[width=0.6\textwidth]{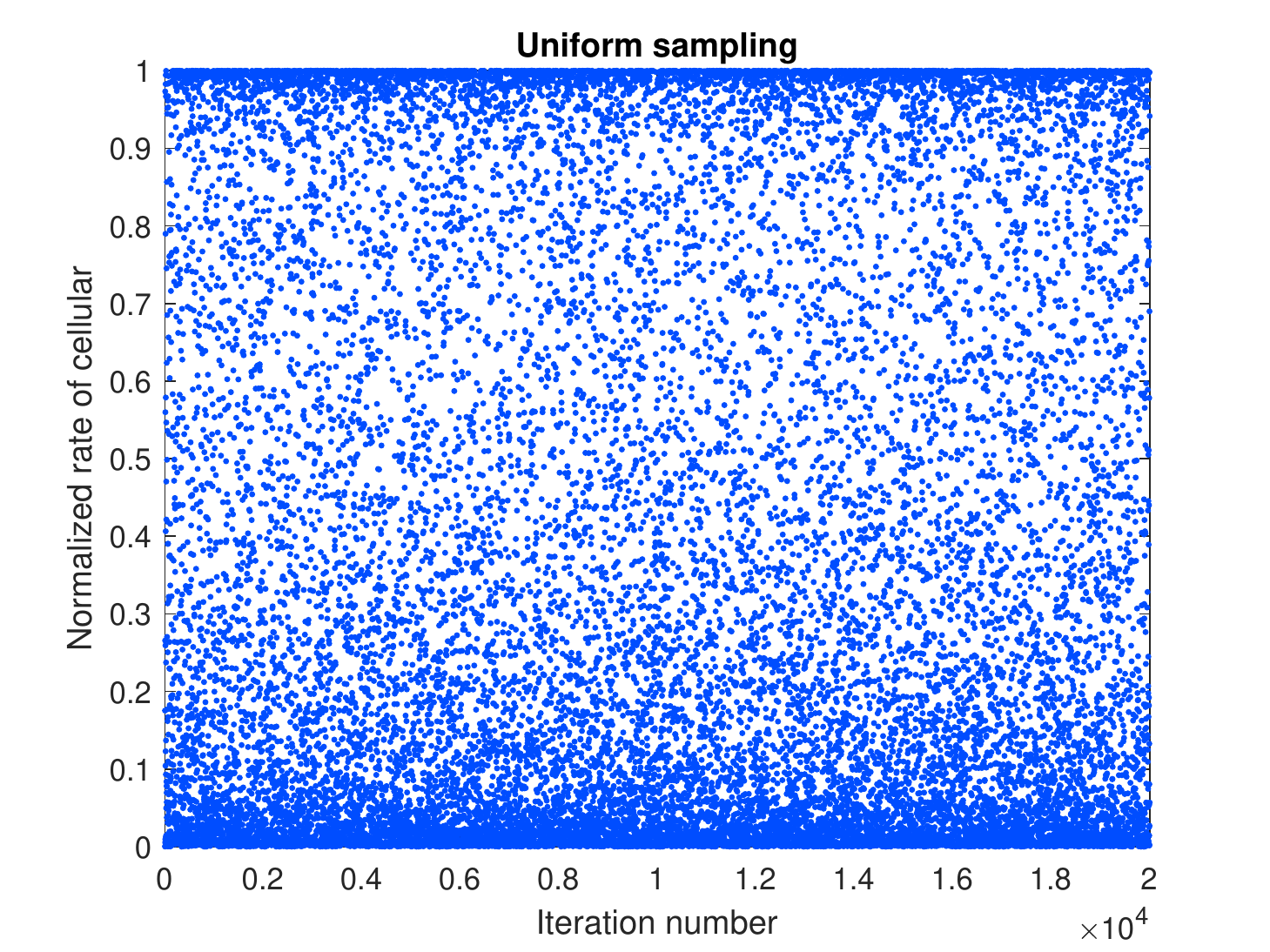}
	\end{center}
	\caption{\small Scatter plot of the normalized rate of the cellular user at each time step of the learning period for uniform sampling.}
	\label{fig:scatteruniform}
\end{figure}

To further evaluate the learning process, in Figs. \ref{fig:scaterlinfullpost} and \ref{fig:scatteruniform}, we show a the scatter plot of the reward that is received at each time steps for linear full posterior and uniform sampling respectively. We can observe from Fig. \ref{fig:scaterlinfullpost} that the density of the scatter plot is more or less uniform at the beginning of the learning period. However, as iteration numbers increase, the contextual MAB is able to select the MTD which lead to higher rewards. We also observe that the linear full posterior is selecting suboptimal MTDs even after it has learned the optimal MTDs. This is due to the noise that is introduced to the sampled posterior distributions and how these algorithm is addressing the exploration vs exploitation dilemma. Therefore, linear full posterior selects the suboptimal MTDs to have a better estimate of their reward distribution. Fig. \ref{fig:scatteruniform} shows that uniform sampling is simply making random selections and then the density of the scatter period during the entire learning period is unchanged. These two scatter plots present how the learning process is able to learn the optimal MTD for the given beamformer without any knowledge on the instantaneous CSI of the links between MTDs and the BS. 

\section{Conclusion}\label{conclutions}
In this paper, we have introduced the idea of OSO for MTCs in capillary networks and presented a contextual bandit learning framework for implementing it with no CSI at the BS for MTDs. First, we have introduced OSO as a method to exploit the natural available null-spaces of the receive beamformers in cellular systems. Second, we have provided a thorough analysis of the properties of the interference from various MTDs on the same beamformer. Moreover, we have given the theoretical analysis of the distributions of the SINR of the cellular user under interference from MTDs and also the distribution of the outage probability. Since implementing OSO requires CSI from all the MTDs at the BS, and, making such an assumption is impractical, we have provided a novel learning method using contextual MABs. Function approximations based on linear regression and neural networks are then used to sample posterior distributions of the well-known Thompson sampling for MABs. Simulation results have shown that, in the regime of massive IoT, it is possible to find an MTD that causes no harmful interference on the BS. Moreover, linear full posterior sampling has shown the best performance which has achieved logarithmic regret and only $10\%$ less total reward compared to the optimal MTD selection policy that requires full CSI of the MTDs. To the best of our knowledge, this is the first paper that introduces the idea of OSO for MTCs and presents a practical method to implement it.

\bibliographystyle{IEEEtran}
\bibliography{IoT_V07}

\begin{thebibliography}{10}
\providecommand{\url}[1]{#1}
\csname url@samestyle\endcsname
\providecommand{\newblock}{\relax}
\providecommand{\bibinfo}[2]{#2}
\providecommand{\BIBentrySTDinterwordspacing}{\spaceskip=0pt\relax}
\providecommand{\BIBentryALTinterwordstretchfactor}{4}
\providecommand{\BIBentryALTinterwordspacing}{\spaceskip=\fontdimen2\font plus
\BIBentryALTinterwordstretchfactor\fontdimen3\font minus
  \fontdimen4\font\relax}
\providecommand{\BIBforeignlanguage}[2]{{%
\expandafter\ifx\csname l@#1\endcsname\relax
\typeout{** WARNING: IEEEtran.bst: No hyphenation pattern has been}%
\typeout{** loaded for the language `#1'. Using the pattern for}%
\typeout{** the default language instead.}%
\else
\language=\csname l@#1\endcsname
\fi
#2}}
\providecommand{\BIBdecl}{\relax}
\BIBdecl

\bibitem{samadVTC}
S.~Ali and N.~Rajatheva, ``Opportunistic scheduling of machine type
  communications as underlay to cellular networks,'' in \emph{Proc. of IEEE
  86th Vehicular Technology Conference (VTC-Fall)}, Toronto, Canada, Sept 2017,
  pp. 1--5.

\bibitem{dawyM2MMagazine}
Z.~Dawy, W.~Saad, A.~Ghosh, J.~G. Andrews, and E.~Yaacoub, ``Toward massive
  machine type cellular communications,'' \emph{IEEE Wireless Communications},
  vol.~24, no.~1, pp. 120--128, February 2017.

\bibitem{walid6G}
W.~Saad, M.~Bennis, and M.~Chen, ``A vision of {6G} wireless systems:
  Applications, trends, technologies, and open research problems,'' \emph{arXiv
  preprint arXiv:1902.10265}, 2019.

\bibitem{IoTin5G}
M.~R. Palattella, M.~Dohler, A.~Grieco, G.~Rizzo, J.~Torsner, T.~Engel, and
  L.~Ladid, ``Internet of things in the {5G} era: Enablers, architecture, and
  business models,'' \emph{IEEE Journal on Selected Areas in Communications},
  vol.~34, no.~3, pp. 510--527, March 2016.

\bibitem{aidin_tcom}
A.~{Ferdowsi} and W.~{Saad}, ``Deep learning for signal authentication and
  security in massive internet-of-things systems,'' \emph{IEEE Transactions on
  Communications}, vol.~67, no.~2, pp. 1371--1387, Feb 2019.

\bibitem{MassiveM2M}
C.~Bockelmann, N.~Pratas, H.~Nikopour, K.~Au, T.~Svensson, C.~Stefanovic,
  P.~Popovski, and A.~Dekorsy, ``Massive machine-type communications in 5{G}:
  physical and {MAC}-layer solutions,'' \emph{IEEE Communications Magazine},
  vol.~54, no.~9, pp. 59--65, September 2016.

\bibitem{surveyofaccess}
M.~T. Islam, A.~e.~M.~Taha, and S.~Akl, ``A survey of access management
  techniques in machine type communications,'' \emph{IEEE Communications
  Magazine}, vol.~52, no.~4, pp. 74--81, April 2014.

\bibitem{RACHM2M2}
A.~Laya, L.~Alonso, and J.~Alonso-Zarate, ``Is the random access channel of lte
  and lte-a suitable for m2m communications? a survey of alternatives.''
  \emph{IEEE Communications Surveys and Tutorials}, vol.~16, no.~1, pp. 4--16,
  2014.

\bibitem{rach_correlated}
A.~E. Kalor, O.~A. Hanna, and P.~Popovski, ``Random access schemes in wireless
  systems with correlated user activity,'' in \emph{Proc. of IEEE 19th
  International Workshop on Signal Processing Advances in Wireless
  Communications (SPAWC)}, Kalamata, Greece, June 2018, pp. 1--5.

\bibitem{nora}
Y.~{Liang}, X.~{Li}, J.~{Zhang}, and Z.~{Ding}, ``Non-orthogonal random access
  for 5{G} networks,'' \emph{IEEE Transactions on Wireless Communications},
  vol.~16, no.~7, pp. 4817--4831, July 2017.

\bibitem{Samad-fastuplinkgrant}
S.~{Ali}, N.~{Rajatheva}, and W.~{Saad}, ``Fast uplink grant for machine type
  communications: Challenges and opportunities,'' \emph{IEEE Communications
  Magazine}, vol.~57, no.~3, pp. 97--103, March 2019.

\bibitem{samad_globecom}
S.~{Ali}, A.~{Ferdowsi}, W.~{Saad}, and N.~{Rajatheva}, ``Sleeping multi-armed
  bandits for fast uplink grant allocation in machine type communications,'' in
  \emph{Proc. of IEEE Globecom Workshops (GC Wkshps)}, Dec 2018, pp. 1--6.

\bibitem{DI_letter}
S.~{Ali}, W.~{Saad}, and N.~{Rajatheva}, ``A directed information learning
  framework for event-driven {M2M} traffic prediction,'' \emph{IEEE
  Communications Letters}, vol.~22, no.~11, pp. 2378--2381, Nov 2018.

\bibitem{NBLTE-M}
R.~Ratasuk, N.~Mangalvedhe, A.~Ghosh, and B.~Vejlgaard, ``Narrowband {LTE-M}
  system for {M2M} communication,'' in \emph{Proc. of IEEE 80th Vehicular
  Technology Conference (VTC2014-Fall)}, Vancouver, Canada, Sept 2014, pp.
  1--5.

\bibitem{3GPP-NB-IoT}
\BIBentryALTinterwordspacing
Y.~E. Wang, X.~Lin, A.~Adhikary, A.~Gr{\"{o}}vlen, Y.~Sui, Y.~W. Blankenship,
  J.~Bergman, and H.~S. Razaghi, ``A primer on 3{GPP} narrowband
  {I}nternet-of-{T}hings ({NB}-{I}o{T}),'' \emph{CoRR}, vol. abs/1606.04171,
  2016. [Online]. Available: \url{http://arxiv.org/abs/1606.04171}
\BIBentrySTDinterwordspacing

\bibitem{CapillaryHamid}
H.~Shariatmadari, P.~Osti, S.~Iraji, and R.~J{\"a}ntti, ``Data aggregation in
  capillary networks for machine-to-machine communications,'' in \emph{Proc. of
  IEEE 26th Annual International Symposium on Personal, Indoor, and Mobile
  Radio Communications (PIMRC)}, Hong Kong, Aug 2015, pp. 2277--2282.

\bibitem{CapillaryMain}
V.~B. Mi{\v{s}}i{\'c}, J.~Mi{\v{s}}i{\'c}, X.~Lin, and D.~Nerandzic,
  ``Capillary machine-to-machine communications: the road ahead,'' in
  \emph{International Conference on Ad-Hoc Networks and Wireless}.\hskip 1em
  plus 0.5em minus 0.4em\relax Springer, 2012, pp. 413--423.

\bibitem{park}
T.~{Park} and W.~{Saad}, ``Distributed learning for low latency machine type
  communication in a massive internet of things,'' \emph{IEEE Internet of
  Things Journal}, pp. 1--1, 2019.

\bibitem{tse2005fundamentals}
D.~Tse and P.~Viswanath, \emph{Fundamentals of wireless communication}.\hskip
  1em plus 0.5em minus 0.4em\relax Cambridge University Press, 2005.

\bibitem{massivemimomtc01}
L.~{Liu} and W.~{Yu}, ``Massive connectivity with massive {MIMO} part {I}:
  Device activity detection and channel estimation,'' \emph{IEEE Transactions
  on Signal Processing}, vol.~66, no.~11, pp. 2933--2946, June 2018.

\bibitem{massivemimomtc02}
------, ``Massive connectivity with massive {MIMO} part ii: Achievable rate
  characterization,'' \emph{IEEE Transactions on Signal Processing}, vol.~66,
  no.~11, pp. 2947--2959, June 2018.

\bibitem{randomBeamforming}
A.~A. {Dowhuszko}, G.~{Corral-Briones}, J.~{H{\"a}m{\"a}l{\"a}inen}, and
  R.~{Wichman}, ``Performance of quantized random beamforming in delay-tolerant
  machine-type communication,'' \emph{IEEE Transactions on Wireless
  Communications}, vol.~15, no.~8, pp. 5664--5680, Aug 2016.

\bibitem{massivebeamforming}
M.~{Goutay}, L.~S. {Cardoso}, and C.~{Goursaud}, ``Massive machine type
  communications uplink traffic: Impact of beamforming at the base station,''
  in \emph{Proc. of 25th International Conference on Telecommunications (ICT)},
  June 2018, pp. 493--497.

\bibitem{shen2009dynamic}
C.~Shen and M.~P. Fitz, ``Dynamic spatial spectrum access with opportunistic
  orthogonalization,'' in \emph{Proc. of 43rd Annual Conference on Information
  Sciences and Systems (CISS),}.\hskip 1em plus 0.5em minus 0.4em\relax
  Baltimore, MD, USA: IEEE, 2009, pp. 600--605.

\bibitem{shen2011opportunistic}
------, ``Opportunistic spatial orthogonalization and its application in fading
  cognitive radio networks,'' \emph{IEEE Journal of Selected Topics in Signal
  Processing}, vol.~5, no.~1, pp. 182--189, 2011.

\bibitem{IntDraining}
H.~Zhou and T.~Ratnarajah, ``A novel interference draining scheme for cognitive
  radio based on interference alignment,'' in \emph{Porc. of IEEE Symposium on
  New Frontiers in Dynamic Spectrum}, 2010, pp. 1--6.

\bibitem{perlaza2010spectrum}
S.~M. Perlaza, N.~Fawaz, S.~Lasaulce, and M.~Debbah, ``From spectrum pooling to
  space pooling: opportunistic interference alignment in {MIMO} cognitive
  networks,'' \emph{IEEE Transactions on Signal Processing}, vol.~58, no.~7,
  pp. 3728--3741, 2010.

\bibitem{junghoon2010}
J.~H. Lee and W.~Choi, ``Opportunistic interference aligned user selection in
  multiuser mimo interference channels,'' in \emph{Proc. of IEEE Global
  Communications Conference (GLOBECOM 2010)}, Dec 2010, pp. 1--5.

\bibitem{lee2011interference}
------, ``Interference alignment by opportunistic user selection in 3-user
  {MIMO} interference channels,'' in \emph{Proc. of IEEE International
  Conference on Communications (ICC)}.\hskip 1em plus 0.5em minus 0.4em\relax
  Kyoto, Japan: IEEE, 2011, pp. 1--5.

\bibitem{sutton1998reinforcement}
R.~S. Sutton and A.~G. Barto, \emph{Reinforcement learning: An
  introduction}.\hskip 1em plus 0.5em minus 0.4em\relax MIT press Cambridge,
  1998, vol.~1, no.~1.

\bibitem{deep_contextual_bandits}
C.~Riquelme, G.~Tucker, and J.~Snoek, ``Deep bayesian bandits showdown: An
  empirical comparison of {B}ayesian deep networks for {T}hompson sampling,''
  Vancouver, Canada, May 2018.

\bibitem{proakis2001digital}
J.~G. Proakis and M.~Salehi, \emph{Digital communications}.\hskip 1em plus
  0.5em minus 0.4em\relax McGraw-hill New York, 2001, vol.~4.

\bibitem{papoulis1965probability}
A.~Papoulis, ``Probability, random variables, and stochastic processes,'' 1965.

\bibitem{integrals_books}
I.~S. Gradshteyn and I.~M. Ryzhik, \emph{Table of integrals, series, and
  products}.\hskip 1em plus 0.5em minus 0.4em\relax Academic press, 2014.

\bibitem{3GPPPathLossModel}
3GPP, ``Radio frequency ({RF}) requirements for {LTE} pico node {B},'' {3rd
  Generation Partnership Project (3GPP)}, Technical Specification (TS) 36.931.

\bibitem{contextual_mab}
\BIBentryALTinterwordspacing
T.~Lu, D.~Pal, and M.~Pal, ``Contextual multi-armed bandits,'' in \emph{Proc.
  of the Thirteenth International Conference on Artificial Intelligence and
  Statistics}, ser. Proceedings of Machine Learning Research, Y.~W. Teh and
  M.~Titterington, Eds., vol.~9.\hskip 1em plus 0.5em minus 0.4em\relax Chia
  Laguna Resort, Sardinia, Italy: PMLR, 13--15 May 2010, pp. 485--492.
  [Online]. Available: \url{http://proceedings.mlr.press/v9/lu10a.html}
\BIBentrySTDinterwordspacing

\bibitem{thompsonSampling}
W.~R. Thompson, ``On the likelihood that one unknown probability exceeds
  another in view of the evidence of two samples,'' \emph{Biometrika}, vol.~25,
  no. 3/4, pp. 285--294, 1933.

\bibitem{droput}
N.~Srivastava, G.~Hinton, A.~Krizhevsky, I.~Sutskever, and R.~Salakhutdinov,
  ``Dropout: a simple way to prevent neural networks from overfitting,''
  \emph{The Journal of Machine Learning Research}, vol.~15, no.~1, pp.
  1929--1958, 2014.

\end{thebibliography}
\end{document}